\documentclass[12pt]{article}
\usepackage{fullpage,epsfig,graphics,amsbsy,amssymb,cancel,bbm}
\usepackage{psfrag,hyperref,color,slashed}
\usepackage{graphicx}
\usepackage{amsfonts,mathdots,dsfont}
\usepackage{amssymb,amsmath,amscd,mathrsfs}
\usepackage{tikz}
\usetikzlibrary{arrows,chains,matrix,positioning,scopes,snakes}

\makeatletter
\tikzset{join/.code=\tikzset{after node path={%
\ifx\tikzchainprevious\pgfutil@empty\else(\tikzchainprevious)%
edge[every join]#1(\tikzchaincurrent)\fi}}}
\makeatother

\tikzset{>=stealth',every on chain/.append style={join},
         every join/.style={->}}
\tikzset{
    >=stealth',
    punkt/.style={
           rectangle,
           rounded corners,
           draw=black, very thick,
           text width=6.5em,
           minimum height=2em,
           text centered},
    pil/.style={
           ->,
           thick,
           shorten <=2pt,
           shorten >=2pt,}
}

\newcommand{\BB}{\mathbb}

\newcommand{\FR}{\mathfrak}

\newcommand{\bea}{\begin{eqnarray}}
\newcommand{\eea}{\end{eqnarray}}
\newcommand{\nn}{\nonumber}
\newcommand{\Tr}{\textrm{Tr}}

\newcommand{\sbullet}{\textrm{\tiny{\textbullet}}}
\newcommand{\udl}{\underline}
\newcommand{\bra}{\langle}
\newcommand{\ket}{\rangle}

\newcommand{\im}{\textrm{Im}\,}

\newcommand{\reeb}{\textrm{\scriptsize{$R$}}}
\newcommand{\sreeb}{\textrm{\tiny{$R$}}}
\newcommand{\gYM}{g_{\textrm{\tiny{$YM$}}}}

\def\gb{\beta}

\def\Gc{\Gamma}
\def\Gd{\Delta}
\def\gd{\delta}
\def\ep{\epsilon}
\def\gt{\theta}

\def\gs{\sigma}

\def\gk{\kappa}
\def\gl{\lambda}
\def\Gl{\Lambda}
\def\Go{\Omega}
\def\go{\omega}

\DeclareMathAlphabet{\mathpzc}{OT1}{pzc}{m}{it}

\newtheorem{theorem}{Theorem}[section]
\newtheorem{lemma}[theorem]{Lemma}

\newenvironment{proof}[1][Proof]{\begin{trivlist}
\item[\hskip \labelsep {\bfseries #1}]}{\end{trivlist}}

\newcommand{\qed}{\nobreak \ifvmode \relax \else
      \ifdim\lastskip<1.5em \hskip-\lastskip
      \hskip1.5em plus0em minus0.5em \fi \nobreak
      \vrule height0.5em width0.5em depth0.00em\fi}

\begin{document}
\renewcommand{\theequation}{\thesection.\arabic{equation}}
\setcounter{page}{0}

\thispagestyle{empty}
\begin{flushright} \small
UUITP-02/14\\
NS-KITP-14-018
 \end{flushright}
\smallskip
\begin{center} \LARGE
{\bf Gluing Nekrasov   partition functions}
 \\[12mm] \normalsize
{\bf  Jian Qiu$^a$, Luigi Tizzano$^b$, Jacob Winding$^b$ and Maxim Zabzine$^b$} \\[8mm]
 {\small\it
${}^a$Math\'ematiques, Universit\'e du Luxembourg,\\
 Campus Kirchberg, G 106,  L-1359 Luxembourg\\
      \vspace{.5cm}
${}^b$Department of Physics and Astronomy,
     Uppsala university,\\
     Box 516,
     SE-75120 Uppsala,
     Sweden\\
   }
\end{center}
\vspace{7mm}
\begin{abstract}
 \noindent In this paper we summarise the localisation calculation of 5D super Yang-Mills on simply connected toric Sasaki-Einstein (SE) manifolds. We show how various aspects of the computation, including the equivariant index, the asymptotic behaviour and the factorisation property are governed by the combinatorial data of the toric geometry.
  We prove that the full perturbative partition function on a simply connected SE manifold corresponding to an ${\tt n}$-gon toric diagram factorises to ${\tt n}$ copies of perturbative Nekrasov partition function. This leads us to conjecture the full partition function as gluing ${\tt n}$ copies of full Nekrasov partition function.
   This work is a generalisation of some earlier computation carried out on $Y^{p,q}$ manifolds, whose moment map cone has a quadrangle and
     our result  is valid for manifolds whose moment map cones have pentagon base, hexagon base, etc. The algorithm we used
      for dealing with general cones may also be of independent
   interest.
    \end{abstract}

\eject
\normalsize
\tableofcontents

\section{Introduction}\label{sec_intro}

Starting from Pestun's work \cite{Pestun:2007rz} there has been an explosion in the applications of localisation technique for supersymmetric gauge theories
 in diverse dimensions.  The calculations were mainly concerned with the evaluation of partition functions and the expectation values of the supersymmetric Wilson loops on (squashed) $S^d$
  and on $S^d \times S^1$, while other geometries were not investigated in detail. However in order to understand the geometrical properties of partition functions,  it
   is important to perform calculations on more general geometries. Five dimensional supersymmetric gauge theories on SE manifolds offer us this possibility and this is the subject of this paper.

 In order to be able to localise 5D supersymmetric Yang-Mills theory we need at least two supersymmetries. Indeed we can construct the supersymmetric
  gauge theory on any simply connected Sasaki-Einstein manifold and the theory preserves two supersymmetries.  In particular there exist very nice examples
   of such manifolds, toric Sasaki-Einstein manifolds (their cones are toric Calabi-Yau manifolds).
 The goal of this work is to present the uniform treatment of localisation calculation for perturbative partition function of 5D supersymmetric Yang-Mills
  on any simply connected toric Sasaki-Einstein manifolds (for the earlier related work in 4D see \cite{Nekrasov:2003vi}). Every such manifold is described in terms of an ${\tt n}$-gon toric diagram
  and topologically corresponds to $({\tt n}-3)(S^2 \times S^3)$which is  $({\tt n}-3)$-fold connected sums of $S^2 \times S^3$ (see proposition 11.4.3 in \cite{BoyerGalicki} or corollary
  5.4 in \cite{2010arXiv1004.2461S}) and they are known as the
Smale manifolds.  This work is a natural continuation and generalisation of the previous calculations for $Y^{p,q}$-spaces \cite{Qiu:2013pta, Qiu:2013aga}.

Let us summarise our main results. Let $X$ be a simply connected toric SE manifold (we will give brief review in section \ref{sec_TSEM} of some features of such manifolds), with moment map cone $C_{\mu}(X)$ defined by
\bea
C_{\mu}(X)=\{\vec r\in\BB{R}^3 | \vec r\cdotp \vec v_i\geq0,~i=1,\cdots\tt n\}~,\nn
\eea
where $\vec v_i$ are the inward pointing normals of the $\tt n$ faces of this cone, for example see Figure \ref{fig_intro}.  The SE condition also implies that there exists a primitive vector $\vec \xi$, such that
\bea
\vec\xi\cdot \vec v_i=1~,~~~\forall i~,\label{gorenstein}
\eea
known as the 1-Gorenstein condition. Up to an $SL(3, \BB{Z})$ rotation, we can make $\vec \xi=[1,0,0]$, we will use this convention throughout the paper.
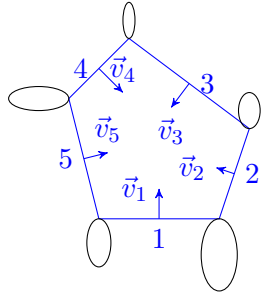
\begin{figure}[h]
\begin{center}
\begin{tikzpicture}[scale=.8]
\draw [-,blue] (-1,-1) -- node[below] {\small$1$} (1,-1) -- node[right] {\small$2$} (1.5,.5) -- node[right] {\small$3$} (-.5,2) -- node[left] {\small$4$} (-1.5,1) -- node[left] {\small$5$} (-1,-1);

\draw (1,-1.6) ellipse (.3 and .6);
\draw (-1,-1.4) ellipse (.2 and .4);
\draw (1.5,0.8) ellipse (.18 and .3);
\draw (-0.5,2.3) ellipse (.1 and .3);
\draw (-2,1) ellipse (.5 and .2);

\draw [->,blue] (0,-1) -- (0,-.5) node[left] {\small$\vec v_1$};
\draw [->,blue] (1.25,-0.25) -- (0.95,-.15) node[left] {\small$\vec v_2$};
\draw [->,blue] (0.5,1.25) -- (0.2,.85) node[below] {\small$\vec v_3$};
\draw [->,blue] (-1,1.5) -- (-0.6,1.1) node[above] {\small$\vec v_4$};
\draw [->,blue] (-1.25,0) -- (-0.85,0.1) node[above] {\small$\vec v_5$};

\end{tikzpicture}

\caption{The polygon base of a polytope cone. Over the interior of the polygon there is a $T^3$ fibre, but over the faces the $T^3$ degenerates into $T^2$, which further degenerate over the vertices to $S^1$, drawn as the circles in the figure. These circles are the only generic closed Reeb orbits.}\label{fig_intro}
\end{center}
\end{figure}

Next let $\vec\reeb$ be a three vector parameterising the Reeb vector field, satisfying the dual cone condition (see the equation (\ref{dual_condition})).
 The perturbative partition function of 5D SYM with a hypermultiplet of mass $m$ and representation $\underline{R}$ on $X$ is given by the matrix model integral
\bea
Z^{pert}
=\int\limits_{\FR{t}}da~e^{-\frac{8\pi^3 r}{\gYM^2}\varrho\,\Tr[a^2]}\cdotp
\frac{{\det}_{adj}' ~  S_3^X (ia; \vec\reeb)} {{\det}_{\underline{R}} ~S_3^X(ia +im +\reeb^1/2; \vec\reeb)}~,\label{Z_pert_sum_intro}
\eea
where we define the generalised triple sine associated to $X$
\bea
 S_3^X (x; \vec\reeb) = \prod\limits_{\vec m\in C_{\mu}(X)\cap\BB{Z}^3}\big(\vec m\cdotp\vec \reeb+x\big)\big(\vec m\cdotp\vec \reeb+\vec\xi\cdotp\vec\reeb-x\big)~,\label{S_3_new}
\eea
where $\vec\xi$ is defined in (\ref{gorenstein}), if we take $\vec\xi=[1,0,0]$ as above, then $\vec\xi\cdotp\vec\reeb$ is simply $\reeb^1$, the first component of $\vec\reeb$. The product is taken over integer points inside the cone
 $C_{\mu}(X)$.  Once we have computed the answer (\ref{Z_pert_sum_intro}), we may  allow $\vec\reeb$ to have complex
 components.
Keeping the real part of $\vec \reeb$ within the dual cone, but giving it a generic imaginary part, then we can factorise the above partition function into
\bea
Z^{pert}
=\int\limits_{\FR{t}}da~e^{-\frac{8\pi^3 r}{\gYM^2}\varrho\,\Tr[a^2]+B_{vec}(ia)+B_{hyp}(ia)}\cdotp
\frac{\prod\limits_{i=1}^{\tt n} {\det}_{adj}'\big(e^{-a\gb_i}\big|e^{i\gb_i\ep_i},e^{i\gb_i\ep_i'}\big)_{\infty}} {\prod\limits_{i=1}^{\tt n}{\det}_{\underline{R}}\big(e^{-(a+m^*)\gb_i}\big|e^{i\gb_i\ep_i},e^{i\gb_i\ep_i'}\big)_{\infty}}~,\label{Z_pert_factor_intro}
\eea
where $m^*=m-\frac{i\reeb^1}{2}$.

We now explain the notations. The index $i$ labels the $\tt n$ closed Reeb orbits in $X$. Each such orbit has circumference $\gb_i$ and the special function $\big(e^{-a\gb_i}\big|e^{i\gb_i\ep_i},e^{i\gb_i\ep_i'}\big)_{\infty}$ defined in Appendix \ref{a-special} by (\ref{special_Narukawa}), is the perturbative part of the Nekrasov partition function on $\BB{C}^2\times S^1$ with equivariant parameters $\ep_i$ and $\ep_i'$.
   The Nekrasov partition function \cite{Nekrasov:2002qd, Nekrasov:2003rj} is defined as counting of states on $\BB{C}^2\times S^1$
\bea
  Z^{full}_{\mathbb{R}^4 \times S^1} = {\rm Tr}_{\cal H} \big( (-1)^{2(j_L + j_R)}e^{- \beta H - i(\epsilon - \epsilon') J_L^3 -i (\epsilon + \epsilon') J_R^3    -i (\epsilon +
   \epsilon') J_I^3}\big)~,\nn\eea
where $H$ is the Hamiltonian, $j_L$ and $j_R$ correspond to the spins under the little group $SO(4)$ and $J_I^3$ is a generator of the R-symmetry group $SU(2)$.
The quantity $m^*$ is the effective mass $m^*=m-i\reeb^1/2$, and $\vec\reeb^1$ here comes from the combination $\vec\reeb\cdotp\vec\xi$ and the choice $\vec\xi=[1,0,0]$. Finally the quantities $\gb,\ep,\ep'$ are defined as follows. Let $i$ label the corner of the intersection of faces $i$ and $i+1$ in Figure \ref{fig_intro}, and choose $\vec n$ such that $\det[\vec v_i,\vec v_{i+1},\vec n]=1$, then
\bea
\frac{\gb_i}{2\pi}=\det[\vec v_i,\vec v_{i+1},\vec \reeb]^{-1}~,~~
\ep_i=\det [\vec\reeb,\vec v_{i+1},\vec n]~,~~\ep_i'=\det[\vec v_i,\vec \reeb,\vec n]~.\label{recipe_bee}
\eea
It is important to stress that the identification of parameters $\ep, \ep'$ is not unique, one may always add to $\ep,\;\ep'$ integer multiples of $2\pi\gb^{-1}$.
The terms $B_{vec}(x)$ and $B_{hyp}(x)$ are polynomials defined in Appendix \ref{a-special}   by (\ref{Bvec}) and  (\ref{Bhyp}).

 The above manner of presenting the factorisation for $Y^{p,q}$
  was used in \cite{Qiu:2013aga}, but it has one drawback, namely the piece we call the perturbative Nekrasov partition function, in particular, the denominator of (\ref{Z_pert_factor_intro}), is not manifestly symmetric under exchange $\udl{R}\to\bar{\udl{R}}$, namely under $a+m\to-a-m$, and it only becomes so when combined with the piece $B_{hyp}$. However this symmetry is expected since the denominator of (\ref{Z_pert_sum_intro}) does possess this symmetry. The reason that the Nekrasov partition function lacks this symmetry is that in the trace, one must let $\ep$ $\ep'$ be complex in order to define the index as a formal power series. In doing so the matter fields of representations $\udl{R}$ and $\udl{\bar R}$ are treated unequally leading to the lack of symmetry. However, we can follow the work \cite{Nieri:2013vba} and factorise also the Bernoulli pieces $B_{vec}$, $B_{hyp}$ and make the symmetry manifest. So a second way of presenting the factorisation is
\bea Z^{pert}=\int\limits_{\FR{t}}
\resizebox{.8\hsize}{!}{$
da~e^{-\frac{8\pi^3 r}{\gYM^2}\varrho\,\Tr_f[a^2]}\cdotp\frac{\prod\limits_{i=1}^{\tt n}\big({\det}_{adj}'\big(e^{-\gb_ia}\big|e^{i\gb_i\ep_i},e^{i\gb_i\ep_i'}\big)_{\infty}\big(a\to -i\reeb^1-a\big)\big)^{1/2}}{\prod\limits_{i=1}^{\tt n}\big(\det_{\underline{R}}\big(e^{-\gb_i(a+m-i\sreeb^1/2)}\big|e^{i\gb_i\ep_i},e^{i\gb_i\ep_i'}\big)_{\infty}\big(a+m\to -a-m\big)\big)^{1/2}}~.\label{Z_pert_factor_I_intro}
$}
\eea
In this way, the partition function is presented as the product of $\tt n$ blocks, each of which corresponds to a copy of partition function associated to $\BB{C}^2\times S^1$, for further investigations of the properties of these blocks see \cite{Nieri:2013yra, Nieri:2013vba}. At this point it is natural to conjecture that the full partition function on $X$ is
 given by same gluing of  $\tt n$-copies of the full Nekrasov partition functions.

The paper is organised as follows: In section \ref{sec_TSEM} we give an overview of the 5D toric SE manifolds, with emphasis on how to read off the geometry from the toric data.
 In section \ref{s-localisation} we present the derivation of full perturbative partition function for any toric simply connected SE manifolds.
  We explain that the answer can be written in two equivalent ways, either using the restricted lattice or using the cone description.
 The result is given in terms of some new special function which is a generalisation of triple sine function.
   Section \ref{sec_act_cal} contains the detailed technical proof of the factorisation of the perturbative partition function in terms of Nekrasov's partition functions
    on $\mathbb{R}^4 \times S^1$. In section \ref{s-summary} we summarise our paper and we conjecture the full nonperturbative answer which contains
     instantons. We also point out some puzzles and open problems in that section.
 The paper is supplemented by two appendices. In appendix \ref{a-special} we collect some basic facts and conventions of the special functions.
 We also prove a property of a special function that we used in the main text. In appendix  \ref{sec_childsplay} we make some comments on the description  of the good cone condition.

\section{Toric Sasaki-Einstein Manifolds}\label{sec_TSEM}

In this section we briefly review some background material concerning the 5D toric Sasaki-Einstein geometry. In particular we concentrate on
 how one may read off from the toric diagram information about the geometry. The reader  may  find similar review in  \cite{Qiu:2013aga} and
  for more detailed exposition one may consult  \cite{2010arXiv1004.2461S,BoyerGalicki}.

Take a manifold $X$ and consider its metric cone $C(X)=X\times\BB{R}^{\gneq0}$ with metric $G=d{\mathfrak{r}}^2+\mathfrak{r}^2g_X$, with $\mathfrak{r}$ being the coordinate of $\BB{R}^{\gneq0}$. If $C_M(X)$ is K\"ahler, then one says that $X$ is a Sasaki manifold, if further $C_M(X)$ is Calabi-Yau, then $X$ is said to be \emph{Sasaki-Einstein} (SE). In particular, its Ricci tensor satisfies
\bea R_{mn}=4g_{mn}\nn\eea
for dimension 5.

Given a Sasaki manifold, one has the metric contact structure, with the Reeb vector field $\vec\reeb$ and contact 1-form $\gk$ given by
\bea
 \reeb=J(\mathfrak{r}\partial_{\mathfrak{r}})~,~~~\gk=i(\bar\partial-\partial)\log \mathfrak{r}~,\nn
 \eea
where $J$ is the complex structure over $C(X)$. If there is
an effective, holomorphic and Hamiltonian action of the torus $T^3$ on the
metric cone $C(X)$, and the Reeb vector field is a linear combination of the torus action, then one says that $X$ is \emph{toric}. Our main examples $S^5$, $Y^{p,q}$-spaces discovered in \cite{Gauntlett:2004yd} and $L^{a,b,c}$-spaces discovered in 
 \cite{Cvetic:2005ft} are all toric SE manifolds. Next we turn to the toric description of these examples and more general toric SE manifolds.

Let $\vec\mu$ be the moment map for the three torus actions, then due to the cone structure on $C(X)$, the image of $\vec\mu$ will also be a cone in $\BB{R}^3$, denoted as $C_{\mu}(X)$. From the cone one can read off almost all information of the manifold, in fact, it is was shown in \cite{2001math......7201L} by Lerman, extending the well-known Delzant construction \cite{Delzant1988}, that from a given \emph{good} cone (definition to come shortly), one can reconstruct the manifold itself. One will see an inkling of how this is done in subsection \ref{sec_long_title}.

Lerman termed a cone to be \emph{good}\footnote{The original formulation is slightly different from the one given here, and since the equivalence does not seem obvious to us, we provide a short proof in the appendix \ref{sec_childsplay}.} if at each intersection of its two adjacent faces $F_i$ and $F_{i+1}$, their inward pointing normals $\vec v_i,\vec v_{i+1}\in\BB{Z}^3$ can be completed into a basis of $\BB{Z}^3$. That is, there exists a third vector $\vec n$ such that $\det[\vec v_i,\vec v_{i+1},\vec n]=1$. A useful
 way of viewing the manifold $C(X)$ is the following: away from the boundary of the moment map cone $C_{\mu}(X)$, one has the torus fibration $T^3\to C(X)\big|_{C_{\mu}(X)^{\circ}}\to C_{\mu}(X)^{\circ}$, where ${}^{\circ}$ means the interior. While at face $i$, the particular torus as singled out by $\vec v_i$ degenerates.

The Reeb vector field is by definition a linear combination of the three torus actions, so one can represent $\reeb$ as a 3-vector $\vec\reeb$. The actual manifold $X$ can be obtained by restricting $C(X)$ to the plane $\vec y\cdotp \vec \reeb=1/2$, and we shall call the intersection
\bea \{\vec y\in C_{\mu}(X)|\vec y\cdotp \vec\reeb=1/2\}=B_{\mu}(X)~,\nn\eea
where $B$ stands for 'base'. This base is a compact polygon iff the 3-vector $\vec\reeb$ is within the \emph{dual cone}
\bea
\vec\reeb=\sum_{i=1}^{\tt n}~\gl_i\vec v_i~,~~\gl_i>0~,~~\forall i~.\label{dual_condition}
\eea
This condition also appears later as the condition for the partition function to converge. From this discussion, one may similarly view $X$ as a torus fibration over $B_{\mu}(X)^{\circ}$ and again at the boundary of $B_{\mu}(X)$, different tori degenerates. An immediate consequence of this view is that the fundamental group of $X$ can be computed as
\bea\pi_1(X)\sim \BB{Z}^3\big/\textrm{span}_{\BB{Z}}\bra \vec v_1,\cdots,\vec v_{\tt n}\ket~.\label{fundamental_gruop}\eea
The meaning of this formula is clear: only those tori that cannot be written as a linear combination of $\vec v_i$ are not contractible.  As a technical remark, if $X$ is simply connected, then it implies that the matrix $[\vec v_1,\cdots,\vec v_{\tt n}]$ can be completed into an $SL({\tt n},Z)$ matrix. Indeed, up to right multiplying by an $SL({\tt n},\BB{Z})$ matrix, one can put $\vec v_{1,2,3}$ into $[1,0,0]$, $[0,1,0]$ and $[0,0,1]$, and the rest is clear.

Furthermore, if $\vec\reeb$ is generic, then the orbit of the Reeb vector field is not closed, except at the corners of $B_{\mu}(X)$, where only one $S^1$ is acting non-trivially. When restricted to a neighbourhood of a corner, the manifold $X$ is a solid torus, i.e. diffeomorphic to $S^1\times\BB{C}^2$, where $S^1$ is the closed Reeb orbit over the corner point, for example see Figure \ref{fig_intro}. But the solid torus is twisted, as one completes a cycle along $S^1$, the two planes also rotate by some angles. The central message of this paper is that to compute the partition function, one need only include one copy of the Nekrasov instanton partition function for each closed Reeb orbit, where the twisting parameters appear as the equivariant parameters of the Nekrasov partition function.

Let us focus now on the neighbourhood of one of the corners of, say, the intersection of face $i$ and face $i+1$, let $\vec n$ be an integer-entry 3-vector such that $\det[\vec n,\vec v_i,\vec v_{i+1}]=1$ (the existence of $\vec n$ is a consequence of the moment map cone being good).
One can then decompose the Reeb vector as a linear combination of $\vec n,~\vec v_i,~\vec v_{i+1}$, that is one decompose the Reeb into one $U(1)$ that remains non-degenerate at the corner, which gives the closed Reeb orbit there, plus two more that degenerate at the same corner, giving the twisting parameter of the solid torus. This reasoning leads to the formulae (\ref{recipe_bee}) for the circumference and twisting parameters.

The Calabi-Yau condition can also be phrased in terms of the data of the cone. Assuming that the number of faces is larger than 3, then it turns out that if there exists an integer vector $\vec\xi$ such that $\vec\xi\cdotp\vec v_i=1,~\forall i$, then $C(X)$ is Calabi-Yau. In fact, it is convenient to choose a basis of the 3-tori so that the first component of $\vec v_i$ is 1 for all $i$ and then $\vec\xi=[1,0,0]$. This property plays a pivotal role in our calculation, in that it allows us to perform a summation within the cone $C_{\mu}(X)$.

Next we give some examples, first the $Y^{p,q}$ space treated in \cite{Qiu:2013aga}, one chooses the four normals to be
\bea \vec v_1=[1,0,0]~,~~\vec v_2=[1,-1,0]~,~~\vec v_3=[1,-2,-p+q]~,~~\vec v_4=[1,-1,-p]~,\label{four_normal_Ypq}\eea
where $p>q>1$ and $\gcd(p,q)=1$.

A generalisation to the $Y^{p,q}$ space is the $L^{a,b,c}$ space, with $d=a+b-c>0$ and $\gcd(a,c)=\gcd(a,d)=\gcd(b,c)=\gcd(b,d)=1$. The four normals are
\bea \vec v_1=[1,c,-bn]~,~~\vec v_2=[1,a,bm],~~\vec v_3=[1,0,1]~,~~\vec v_4=[1,0,0]~,\label{in_normal_L_abc_II}\eea
where $m,n$ are chosen so that $mc+na=1$. The metric cone $C(L^{a,b,c})$ can be constructed as
 K\"ahler quotient of $\mathbb{C}^4$ with $U(1)$ with the charges $(a,b, -c , -a-b+c)$.

As an example of a pentagon toric cone, one has the so called $X^{p,q}$, $p>q>0$ space, whose normals are
\bea
\vec v_1=[1,0,0]~,~~\vec v_2=[1,1,0]~,~~\vec v_3=[1,0,p]~,~~\vec v_4=[1,-1,p+q]~,~~\vec v_5=[1,-1,p+q-1]~.\nn\eea
The metric cone of this space can be constructed from the K\"ahler quotient of $\BB{C}^5$ with respect to
  two $U(1)$'s of charge $[1,0,-1,p,-p]$ and $[1,-1,1,q-1,-q]$.

For the general case we have the following alternative description of $C(X)$ as K\"ahler quotient
\bea
C(X)= \mathbb{C}^{{\tt n}}// U(1)^{\tt n-3}~,\label{Kah-quotient}
\eea
where every $U(1)$ acts on  $\mathbb{C}^{{\tt n}}$ with the charges
 $\vec{Q}_a= (Q_a^1, ..., Q^i_a, ... , Q_a^{\tt n})$ and to ensure CY condition we require $\sum\limits_{i=1}^{\tt n} Q_a^i=0$. If $X$ is simply connected SE manifold then
  one can pick vectors $\vec u_1, \vec u_2, \vec u_3\in\BB{Z}^{\tt n}$ such that $A=[\vec u_1, \vec u_2 ,\vec u_3 ,\vec Q_1,\cdots,\vec Q_{\tt n-3}]$ forms an $SL({\tt n},\BB{Z})$ matrix.  The vectors  $v_i^a$ are defined as the first $3$ rows of $A^{-1}$. We denote these vectors by $\vec v_1,\cdots,\vec v_{\tt n}$, i.e. $\vec v_i$ are $3$-vectors, and the conditions $\sum\limits_{a=1}^{3} v_i^am_a\geq0$ describe a cone inside $\BB{R}^{3}$. This cone is none other than the moment map cone of the SE manifold, and the $\vec v_i$'s are the inward pointing normals (but not necessarily in the correct order).

In what follows we concentrate only on simply connected SE toric manifolds which topologically correspond to $({\tt n}-3)(S^2 \times S^3)$,
namely $({\tt n}-3)$ connected sum of $S^2 \times S^3$.
 We will make a few comments
 about non-simply connected SE manifolds in the last section \ref{s-summary}.


\section{Localisation of 5D SYM}\label{s-localisation}

In this section we sketch briefly the actual localisation calculation. Our presentation is the generalisation of
 the previous works \cite{Kallen:2012cs, Kallen:2012va, Qiu:2013pta} to the case of general simply connected toric SE manifolds.  We also discuss two different representations of the answer.

\subsection{Localisation calculation}

In \cite{HosomichiSeongTerashima} the SYM theory coupled to matter on the round $S^5$ was written down. Due to the SE structure over $S^5$, one can find a pair of normalised Killing spinors $\xi_{1,2}$, such that the bilinear $\xi_1\Gc^m\xi_2$ is proportional to the Reeb vector field $\reeb^m$ on $S^5$. The two Killing spinors will pick out a particular susy charge called $\gd$ that satisfies the key relation
\bea \gd^2=-iL_{\sreeb}+G~,~~~\textrm{for the vector multiplet} \label{susy_sq_vec}\\
 \gd^2=-iL^s_{\sreeb}+G~,~~~\textrm{for the hypermultiplet} \label{susy_sq_hyp}\eea
where $G$ stands for gauge transformation and $L_{\sreeb}$ ($L^s_{\sreeb}$) is the (spinor) Lie derivative.

It turns out that a change of variables (which again involves the Killing spinors) allows us to formulate the vector multiplet in terms of differential forms, and the only feature that is required from the geometry is the metric contact structure. This was called the twisted SYM in \cite{Kallen:2012cs}, and the susy complex was called the \emph{cohomological complex}.
Using the algebra (\ref{susy_sq_vec}), the path integral localises onto the so called contact instanton configurations, and one needs to integrate over the Gaussian fluctuations around such configurations. To calculate the full partition function from first principles appears to be hard at the moment. However
 the expansion around zero connection configuration is doable and one obtains the perturbative partition function as a matrix model.
  Furthermore, since the actual SE metric is not required once we pass to the cohomological complex formulation, we can consider the partition function for the deformed Reeb vector field, i.e. the squashed five sphere. Equivalently, one can turn on extra background gauge fields and put the original SYM theory directly on a squashed $S^5$ and perform the computation from there, see \cite{Imamura:2012xg, Imamura:2012bm} (see also \cite{Kim:2012ava, Lockhart:2012vp, Kim:2012qf}), but the work load is considerably heavier this way around.
For the hypermultiplet, one would need in principle the SE metric, however, once the result is obtained, it is obvious how to generalise it to the squashed sphere.

Much of the story can be repeated for an infinite class of simply connected SE manifolds $Y^{p,q}$. The simply connectedness is there to ensure that the zero instanton configuration actually corresponds to the trivial connection. The calculation was completed in \cite{Qiu:2013pta} for the $Y^{p,q}$ manifolds. The main technical aspect of the calculation, the computation of an equivariant index, relies on using the known index structure on $S^3\times S^3$, and imposing a lattice constraint, as we shall review shortly.
This calculation carries over to the $L^{a,b,c}$ as a straightforward generalisation. But for toric SE manifolds with a more complicated moment map cone, the method used there gets cumbersome, and it is more systematic to employ the fixed point theorem \cite{Ellip_Ope_Cpct_Grp}, presented in the appendix of \cite{Qiu:2013pta}.

Now our goal is to generalise the result from \cite{Qiu:2013pta} to any simply connected toric SE manifold. Following the logic presented in  \cite{Kallen:2012cs, Kallen:2012va, Qiu:2013pta}  for any simply connected SE manifold the perturbative partition function of $N=1$ SYM with a hypermultiplet in representation $\underline{R}$ and mass $m$ is written as the superdeterminant of the two operators in (\ref{susy_sq_vec}) and (\ref{susy_sq_hyp}), taken over the
 $\Go^{0,\sbullet}_H$-complex\footnote{In writing this expression we skipped a few technical steps, in particular, the integral of $a$ within the Lie-algebra of the gauge group is written as the integral over Cartan times a determinant factor, which combines with the contribution from the ghost sector to give this neat expression (\ref{matrix-superdet}).}
\bea
 Z^{pert}
=\int\limits_{\FR{t}}da~e^{-\frac{8\pi^3 r}{\gYM^2}\varrho\,\Tr[a^2] }\cdotp
\frac{{\det}_{adj}'\,\textrm{sdet}_{\Go^{0,\sbullet}_H}(-irL_{\sreeb}-ia)}{{\det}_{\underline{R}}\,\textrm{sdet}_{\Go^{0,\sbullet}_H}(-irL^s_{\sreeb}-ia-im)}~,\label{matrix-superdet}
\eea
where $r$  is a parameter controlling the overall size of $X$, $\varrho$ is the squashed volume of $X$ normalised against $\textrm{Vol}_{S^5}=\pi^3$. The actual non-trivial
 calculation  is centred around the explicit evaluation of superdeterminants in (\ref{matrix-superdet}).

 There exists  different methods to evaluate the superdeterminants in (\ref{matrix-superdet}).
In this section we shall use the method due to Schmude \cite{Schmude:2014lfa}, see also
 \cite{Schmude:2013dua,Eager:2012hx}.
Sasaki manifolds have a transverse K\"ahler structure, that is, one can write the 5D metric as
\bea
g=\gk\otimes \gk+g_H\nn
\eea
with $g_H$ being a \emph{local} K\"ahler metric, see subsection 1.2 of \cite{2010arXiv1004.2461S}. Thus one has the complex of horizontal $(0,i)$ forms, with $i=0,1,2$. By projecting the de Rham differential to its component that increase the degree $(0,i)\to (0,i+1)$, we define a transverse Dolbeault differential $\bar\partial_H$, whose cohomology is called the Kohn-Rossi (KR) cohomology.
The various fields in the vector-, hypermultiplet can be reduced to the horizontal $(0,i)$ forms using Fierz identity, and fit nicely into the $\bar\partial_H$ complex, see either \cite{Kallen:2012cs} or \cite{Kallen:2012va} for details.

As is standard for localisation, only those modes that are in the KR cohomology (which we denote simply as $H^{0,\sbullet}$) make a net contribution to the superdeterminant, so the final answer is
\bea
Z^{pert}
=\int\limits_{\FR{t}}da~e^{-\frac{8\pi^3 r}{\gYM^2}\varrho\,\Tr[a^2] }\cdotp
\frac{{\det}_{adj}'\,\textrm{sdet}_{H^{0,\sbullet}}(-irL_{\sreeb}-ia)}{{\det}_{\underline{R}}\,\textrm{sdet}_{H^{0,\sbullet}}(-irL^s_{\sreeb}-ia-im)}~.\label{Z_pert_review}
\eea
In writing $Z^{pert}$, we ignore some possible ($a$ independent) phases coming from the determinant factors.

It was pointed out by Schmude that the KR cohomology can be reduced to $H^0({\cal O}(C(X)))$, with ${\cal O}(C(X))$ being the sheaf of holomorphic functions on the metric cone of $X$. We will go over this argument here.
Since the Reeb is Killing with respect to the metric, the operator $L_{\sreeb}$ will commute with $\bar\partial_H$, and we can analyse the cohomology of $\bar\partial_H$ with definitive $-iL_{\sreeb}$ eigenvalue, say, $\zeta$ (this eigenvalue is the R-charge). Now one can find a map relating the horizontal $(0,i)$ forms on $X$ to those on the metric cone $C(X)$. Assuming $X$ is embedded in $C(X)$ at $\mathfrak{r}=1$, the Dolbeault differential $\bar\partial$ on $C(X)$ is related to $\bar\partial_H$ in local coordinates by
\bea \bar\partial=\bar\partial_H+\frac12(d\log \mathfrak{r}-i\gk)(\mathfrak{r}\partial_{\mathfrak{r}}
 +i\partial_{\gt})~,\label{bardel_bardel_H}\eea
where $\gt$ is the local coordinate such that $\partial_{\gt}$ is the Reeb vector. Assuming that $\go\in \Go_H^{0,i}$ has eigenvalue $\zeta$ under $-iL_{\sreeb}$, then we can extend it to a form on the $C(X)$ as
 \bea
ext:~~\go\to \go \mathfrak{r}^{\zeta}~,\nn
\eea
the extension makes sense since the point $\mathfrak{r}=0$ is removed. Furthermore if $\go$ is closed (exact) under $\bar\partial_H$ then $\go \mathfrak{r}^{\zeta}$ is closed (exact) under $\bar\partial$, thus the extension induces a map of the corresponding cohomology.
Conversely a $(0,i)$-form on $C(X)$ can be restricted to $X$, the map $ext$ composed with the restriction gives the identity map $res\circ ext=1$. So we see that the induced map on 
cohomology induced by $ext$ must be injective. This implies immediately that $H^{0,1}(X)$ is zero since $H^1({\cal O}(C(X)))=0$. For zeroth cohomology $H^{0,0}(X)$, since there are
 no exact forms, and if a function $f$ is holomorphic on $C(X)$, its restriction to $\FR{r}=1$ is non-zero, as can be seen from (\ref{bardel_bardel_H}) (for example, one can expand $f$ into Laurent series of $\FR{r}$ and modes of different power in $\FR{r}$ must have different $\gt$ eigenvalue and hence cannot cancel out at $\FR{r}=1$). 
So we actually get a bijection
\bea ext:~H^{0,0}(X)\simeq H^0({\cal O}(C(X)))~,\nn\eea
For the $(0,2)$ forms, one can use the holomorphic volume form $\Go$ on $C(X)$ to construct a pairing between $(0,0)$ and $(0,2)$ forms. Since $\Go$ is a top holmorphic form, it is closed, and its restriction to $X$ (also denoted as $\Go$) is closed as well. The restriction of $\Go$ has the property that $\iota_{\sreeb}\Go=1$ and its horizontal component is in $\Go_H^{2,0}(X)$. From these properties, we see that the integration
\bea
\bra f,\go\ket=\int\limits_X\Go f\go~,~~f\in\Go^{0,0}(X)~,~\go\in\Go_H^{0,2}(X)\nn
\eea
is a non-degenerate paring. It is also a non-degenerate pairing between $H^{0,0}(X)$ and $H^{0,2}(X)$, to see this, let $f\in H^{0,0}(X)$, and $\go=\bar\partial_H\zeta,~\zeta\in \Go_H^{0,1}$, then
\bea
\bra f,\bar\partial_H\zeta\ket=\int\limits_X\Go f\bar\partial_H\zeta=\int\limits_X\Go fd\zeta=\int\limits_X\Go df\zeta=\int\limits_X\Go \bar\partial_Hf\,\zeta=0~.\nn\eea
From these considerations $H^{0,2}(X)\simeq (H^{0,0}(X))^*$.

To summarise, to obtain the KR cohomology for our specific problem, it suffices to compute $H^0({\cal O}(C(X)))$, i.e. the holomorphic functions on $C(X)$. But the latter object has a combinatorial description, one simply enumerates the integral points within the moment map cone (this follows almost directly from the definition of a toric K\"ahler manifold)
and each such point gives a holomorphic function on $C(X)$.
What is more, the three coordinates of these points give the charges of these functions under the three $U(1)$'s. In particular, one can also read off their $L_{\sreeb}$ eigenvalue. To figure out the $U(1)$ charges of $H^{0,2}$ groups, one needs to get the charges of $\Go$. To do this, let $\vec\xi$ be the 3-vector such that $\vec\xi\cdotp \vec v_i=1$ (see (\ref{gorenstein})), then the charges of $\Go$ are the 3-components of $\vec \xi$, which is also a standard fact of toric geometry. Then in particular, the R-charge of $\Go$ (the $-iL_{\sreeb}$ eigenvalue) is $\vec\xi\cdotp\vec\reeb$. In all of the examples given earlier, this vector $\vec\xi$ is chosen to be $[1,0,0]$, and so the R-charge is $\reeb^1$.

With these preparation, one can write the superdeterminant in (\ref{Z_pert_review}) as
\bea
\textrm{sdet}_{H^{0,\sbullet}}(-irL_{\sreeb}+x)=\prod_{\vec n\in C_{\mu}(X)\cap\BB{Z}^3}\big(\vec n\cdotp\vec\reeb+x\big)\big(\vec n\cdotp\vec\reeb-x+\vec\xi\cdotp\vec\reeb\big)= S_3^X(x; \vec\reeb)~,\label{triple-anotherdef}
\eea
where the second factor comes from $H^{0,2}$ and we have as usual discarded overall multiplicative constants. The superdeterminant of $-iL_{\sreeb}^s$ is similar, one makes a shift $x\to x+\reeb^1/2$, which originates from expressing $L^s_{\sreeb}$ in terms of $L_{\sreeb}$ \cite{Qiu:2013pta}.
  In (\ref{triple-anotherdef}) we defined a new special function $S_3^X(x; \vec\reeb)$ associated with the moment map cone of any 5D simply connected toric SE manifold $X$. This function is a generalisation of the usual triple sine function, since by taking $X=S^5$, whose moment map cone $C_{\mu}(S^5)=\BB{R}_{\geq0}^3$, one recovers the definition of the standard triple sine function (\ref{mult_sine}).

 To summarise, the perturbative partition function of $N=1$ supersymmetric Yang-Mills over a 5D simply connected toric SE manifold, with hypermultiplet in representation $\udl{R}$ is given by
\bea
Z^{pert}
=\int\limits_{\FR{t}}da~e^{-\frac{8\pi^3 r}{\gYM^2}\varrho\,\Tr[a^2]}\cdotp
\frac{{\det}_{adj}' ~  S_3^X (ia; \vec\reeb)} {{\det}_{\underline{R}} ~S_3^X(ia +im +\reeb^1/2; \vec\reeb)}~,\label{Z_pert_text}\eea
where we have fixed $\vec\xi=[1,0,0]$.
Let us make a couple of concluding remarks. In the setup of the supersymmetric Yang-Mills, especially for the hypermultiplet, we have used the SE metric, and so in particular, the classical action evaluated at the localisation locus (the term in the exponent above) should be $-8r\gYM^{-2}\textrm{Vol}_{X_{SE}}\,\Tr[a^2]$ with $\textrm{Vol}_{X_{SE}}$ computed with the SE metric.
However the superdeterminant of the operator $L_{\sreeb}$ may
be computed for a Reeb being any combination of the three $U(1)$'s, provided $\vec\reeb$ is in the dual cone. These Reebs do not give rise to an SE metric,
and so we have also replaced the volume factor in the exponent by the squashed volume
\bea \textrm{Vol}_X=\varrho\pi^3~.\nn\eea
For a self-contained justification of this replacement, one should set up the supersymmetric Yang-Mills with a general Reeb, which then entails turning on an extra background connection to maintain supersymmetry.  Alternatively we may adopt the cohomological complex as the starting point, as in \cite{Kallen:2012cs}, and then this classical term appears as $\int\limits_X\gk d\gk^2\Tr[\gs^2]$, which is a supersymmetry completion of the Chern-Simons like observable $\int\limits_X\gk FF$. Since the integral of $1/2\gk d\gk d\gk$ leads to the squashed volume, it is natural to make the replacement as we did above.

Using these arguments  the answer given above should be regarded as a general equivariant answer. This is valuable since the equivariant parameters that enter into the $\vec\reeb$ can tell us about how the geometry of the underlying manifold affect the partition function, the factorisation property studied in this paper is just one such instance.
One can also study certain degeneration limits by giving these parameters special values. This will be investigated further in another publication.

\subsection{Relation between the restricted lattice and the cone descriptions}\label{sec_long_title}

In this section we show that the original presentation of the partition function in \cite{Qiu:2013pta} in terms of a constrained lattice is equivalent to the cone description given above. For those familiar with toric geometry, the equivalence is probably quite obvious and he may skip to the next section.

In \cite{Qiu:2013pta} the superdeterminant for $Y^{p,q}$ is given in terms of a generalised triple sine function, which is
defined through the $\zeta$-function regularised infinite product on a lattice
\bea
\begin{split}\label{triplegensine-def}
&\textrm{sdet}_{H^{0,\sbullet}}(-irL_{\sreeb}+x)=S_3^{\Gl_{(p,q)}}(x|  \omega_1, \omega_2, \omega_3, \omega_4) \\
&=\prod_{(i,j,k,l) \in \Gl^+_{(p,q)}}\Big(i \omega_1+j
\omega_2+k\omega_3+l\omega_4 +x\Big)\Big(i\omega_1+j\omega_2+k\omega_3+l\omega_4+\textrm{\small$\sum$}\go_i -x\Big)~,
\end{split}
\eea
  where the lattice $\Lambda^+_{(p,q)}$ is defined as
  \bea
  \Gl^+_{(p,q)}=\big\{i,j,k,l\in\BB{Z}_{\geq0}\;|\;i(p+q)+j(p-q) -kp- lp=0\big\}~,\label{lattice-1}
\eea
 and $\omega_1, \omega_2, \omega_3, \omega_4$ are equivariant parameters which are related to the Reeb vector as follows
\bea
\reeb^1=\sum\go_i~,~~~\reeb^2=-\go_1-\go_2-2\go_4~,~~~\reeb^3=-p\go_2+(q-p)\go_4~,\label{omega_reeb}
\eea
If one replaces the constraint (\ref{lattice-1}) for the lattice by
\bea
  \Gl^+_{(a,b,c)}=\big\{i,j,k,l\in\BB{Z}_{\geq0}\;|\;i a+jb -kc- l (a+b-c)=0\big\}~,\label{lattice-12}
\eea
one obtains the generalised triple sine function $S_3^{\Gl_{(a,b,c)}}$ that gives the perturbative partition function for the $L^{a,b,c}$ manifolds.
 Next we shall see  how to get these relations for a general toric SE manifold.

In general situation for any toric simply connected $X$ we assume that we have a lattice  of $\BB{Z}_{\geq0}^{\tt n}$, obeying ${\tt n} - 3>0$ constraints
\bea
\Gl^+=\{n_i\geq0,~i=1,\cdots,{\tt n}\,|\, \sum_{i=1}^{\tt n} Q_a^i n_i=0,~~a=1,\cdots,{\tt n}-3\}~.\label{restricted_lattic}
\eea
 The charges $Q_a^i $ are the same as in the description of $C(X)$ as K\"ahler quotient in (\ref{Kah-quotient}).  Introducing the squashing parameters
$\vec{\omega}=(\omega_1, ... , \omega_{\tt n})$ we define the generalised triple sine associated with the lattice $\Gl^+$
\bea
 S^{\Gl}_3(x; \vec{\omega})= \prod\limits_{\vec{n} \in \Gl^+} \big(\vec n\cdotp\vec\omega+x\big)\big(\vec n\cdotp\vec\omega-x+\sum\limits_{i=1}^{\tt n} \omega_i \big)~.\label{triple-onedef}
\eea
This was how the result was presented in \cite{Qiu:2013pta}, next we show that this is equivalent with the function $S_3^X$ define in (\ref{S_3_new}), which is a more intrinsic description.

For $X$ simply connected, we can pick the basis vectors $\vec u_1, \vec u_2, \vec u_3\in\BB{Z}^3$ such that $A=[\vec u_1, \vec u_2,\vec u_3,\vec Q_1,\cdots,\vec Q_{{\tt n}-3}]$ forms an $SL({\tt n},\BB{Z})$ matrix. Apply $A$ to the lattice $\Gl^+$, then the $\tt n$ conditions $n_i\geq0$ in (\ref{restricted_lattic}) turn into $\sum\limits_{a=1}^{3}v_i^am_a\geq0$, where $v_i^a$ are the first $3$ rows of $A^{-1}$. We denote these by $\vec v_1,\cdots,\vec v_{\tt n}$, i.e. $\vec v_i$ are $3$-vectors, and the conditions $\sum\limits_{a=1}^{3}v_i^am_a\geq0$ describes a cone inside $\BB{R}^{3}$.

As an illustration, take the lattice (\ref{lattice-1}), then $\vec Q$ is the 4-vector $[-p-q,q-p,p,p]$. One can complete it into an $SL(4,\BB{Z})$ matrix
\bea A=\left(
         \begin{array}{cccc}
           0 & -1 & a & -p-q \\
           0 & 0 & -a-2b & q-p \\
           1 & 1 & b & p \\
           0 & 0 & b & p \\
         \end{array}
       \right)~,~~~(a+b)p+bq=1~.\nn\eea
Its inverse is
\bea A^{-1}=\left(
         \begin{array}{cccc}
           1 & 1 & 1 & 1 \\
           -1 & -1 & 0 & -2 \\
           0 & -p & 0 & q-p \\
           0 & b & 0 & a+2b \\
         \end{array}\right)~,\nn\eea
and from the first three rows of $A^{-1}$ one finds the four inward normals given in (\ref{four_normal_Ypq}). Also the first three rows give the relation of the Reeb vector with $\go_i$ as in (\ref{omega_reeb}).
The above process is reversible if the moment map cone satisfies certain constraints, and the  construction
 mirrors the Delzant and Lerman constructions \cite{Delzant1988,2001math......7201L}, that is, by embedding a cone in $\BB{R}^3$ into $\BB{R}^{\tt n}$ as the intersection of $\tt n-3$ hyperplanes (whose normals are the $\vec Q_a$'s), one can present the original manifold as a K\"ahler quotient of $\BB{C}^{\tt n}$.

Continuing with our manipulation of the lattice, we let $\vec \go$ be an $\tt n$-vector, by inserting $AA^{-1}$ into $\vec n\cdotp\vec \go$, we see that the summation over the constrained lattice can be written as
\bea
\sum_{\Gl^+}\vec n\cdot\vec\go~~=\sum_{m_a\in C_{\mu}(X)\cap\BB{Z}^3}\sum_{a=1}^{3} m_a(A^{-1}\vec\go)_a~. \label{equiv_lattic_cone}
\eea
Thus we have proved the equality of the two products
\bea \prod\limits_{\vec{n} \in \Gl^+}\big(\vec n\cdotp\vec\omega+x\big)~~= \prod\limits_{\vec m\in C_{\mu}(X)\cap\BB{Z}^3}\big(\vec m\cdotp\vec \reeb+x\big)~,~~\textrm{where}~~\reeb_a=(A^{-1}\vec\go)_a~.\nn\eea
Also notice that since $\vec\xi\cdotp\vec v_i=1,~\forall i$, and that $[\vec v_1,\cdots,\vec v_{\tt n}]$ constitutes the first three rows of $A^{-1}$, so the quantity $\vec\xi\cdotp\vec\reeb$ can be written as
\bea \vec\xi\cdotp\vec\reeb=\sum_{a=1}^3\xi_a\reeb_a=\sum_{a=1}^3\xi_a(A^{-1}\vec\go)_a=\sum_{i=1}^{\tt n}\go_i~.\nn\eea
By comparing the definition (\ref{triple-onedef}) and (\ref{S_3_new}) of $S_3^{\Gl}(x,\vec\go)$ and $S_3^{X}(x,\vec\reeb)$, we get the equality
\bea
S^{\Gl}_3(x; \vec{\omega})= S^X_3(x; \vec\reeb)~,\nn\eea
and also the equivalence between the constrained lattice presentation and the cone representation.
%

Next we shall work with a general good cone that corresponds to a 5D simply connected toric SE manifold. Assume that the moment map cone has ${\tt n}\geq 4$ faces, and that the normals are chosen so that their first component is 1.
The perturbative partition is given in (\ref{Z_pert_text}) and
our central task is to evaluate the two products
\bea I:~~\prod_{\vec m\in C_{\mu}(X)\cap\BB{Z}^3}\Big(\vec m\cdotp\vec \reeb+x\Big)~,\label{sum_in_cone_I}\\
II:~~\prod_{\vec m\in C_{\mu}(X)\cap\BB{Z}^3}\Big(\vec m\cdotp\vec \reeb-x+\reeb^1\Big)~.\label{sum_in_cone_II}\eea

\section{Derivation of Factorisation}\label{sec_act_cal}

\subsection{Conversion to the Triple Sine Functions}
Since the real part of the Reeb vector $\vec\reeb$ is assumed to be within the dual cone, and that $x$ has a small but positive real part, the real part of the factors in (\ref{sum_in_cone_I}) is bounded away from zero and tends to infinity, so one can use $\zeta$-function regularisation to make sense of the infinite product. Bearing this in mind, one can treat the infinite product at its face value, and do the usual manipulations.

The product or summation over the integral points within the cone is investigated in \cite{Benvenuti:2006qr} through subdividing the cone into smaller portions. We will use similar strategies that work for any cone that gives rise to simply connected toric SE manifolds. We fix the inward normals of the cone to be $\vec v_i=[1,-\vec L_i],~i=1,\cdots,\tt n$ for some two vectors $\vec L_i=[L_i^2,L_i^3]$.
\begin{figure}[h]
\begin{center}
\begin{tikzpicture}[scale=1]
\draw [->] (-2,0) -- (2,0) node [right] {\small$y$};
\draw [->] (0,-2) -- (0,2) node [left] {\small$z$};

\draw [-,blue] (-1,-1) -- node[below] {\small$1$} (1,-1) -- node[right] {\small$2$} (1.5,.5) -- node[right] {\small$3$} (-.5,2) -- node[left] {\small$4$} (-1.5,1) -- node[left] {\small$5$} (-1,-1);

\draw (1,-1.6) ellipse (.3 and .6);
\draw (-1,-1.4) ellipse (.2 and .4);
\draw (1.5,0.8) ellipse (.18 and .3);
\draw (-0.5,2.3) ellipse (.1 and .3);
\draw (-2,1) ellipse (.5 and .2);

\node at (1.9,-1.6) {\small$(\gb,\ep,\ep')$};
\end{tikzpicture}
\begin{tikzpicture}[scale=1]
\draw [->,white] (-2,0) -- (2,0) node [right] {\small$y$};
\draw [->,white] (0,-2) -- (0,2) node [left] {\small$z$};

\draw [->,blue] (0,0) -- (0,1.5) node [above] {\small$v_1$};
\draw [->,blue] (0,0) -- (-1.2,0.4) node [above] {\small$v_2$};
\draw [->,blue] (0,0) -- (-0.9,-1.2) node [left] {\small$v_3$};
\draw [->,blue] (0,0) -- (1,-1) node [below] {\small$v_4$};
\draw [->,blue] (0,0) -- (1.6,0.4) node [above] {\small$v_5$};
\end{tikzpicture}
\caption{The polytope cone, projected onto the plane $y=1$, depending on the specific case, one of the faces may move off to infinity, that its two neighbouring faces turn parallel. The circles represent the closed Reeb orbits. The right panel is the inward pointing normals of the cone.}\label{fig_toric_cone}
\end{center}
\end{figure}
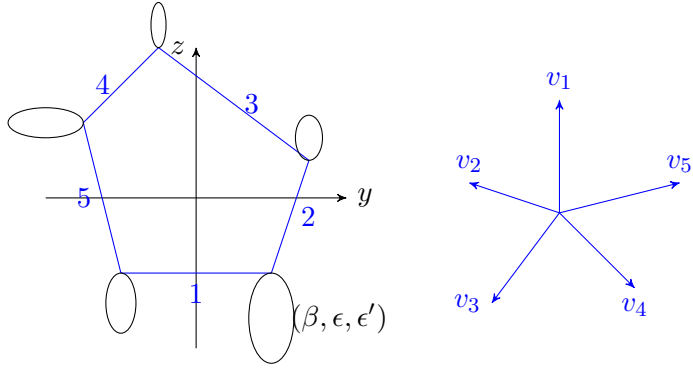
From the constraint $\vec v_i\cdotp \vec m\geq0$, the limit of $m_1$ is $\infty>m_1\geq L^2_im_2+L_i^3m_3$, which changes as $i$ increments. So we need to divide the $m_2$-$m_3$ plane into $\tt n$ (5, in the case of Figure \ref{fig_toric_cone}) wedges, one for each edge, and we get the picture of Figure \ref{fig_sub_div}. So in $W_i$ the lower limit of $m_1$ is $m_1\geq L_i^2m_2+L_i^3m_3$.
\begin{figure}[h]
\begin{center}
\begin{tikzpicture}[scale=1]

\draw [dashed,blue] (0,0) -- (-2,-2);
\draw [dashed,blue] (0,0) -- (2,-2);
\draw [dashed,blue] (0,0) -- (3,1);
\draw [dashed,blue] (0,0) -- (-.6,2.4);
\draw [dashed,blue] (0,0) -- (-2.4,1.6);

\node at (0,-1) {\small$W_1$};
\node at (1.25,-0.25) {\small$W_2$};
\node at (.5,1.25) {\small$W_3$};
\node at (-1,1.5) {\small$W_4$};
\node at (-1.25,0) {\small$W_5$};
\end{tikzpicture}
\caption{The division of the $m_2$-$m_3$ plane, each $W$ corresponds to a face of the moment map cone.}\label{fig_sub_div}
\end{center}
\end{figure}
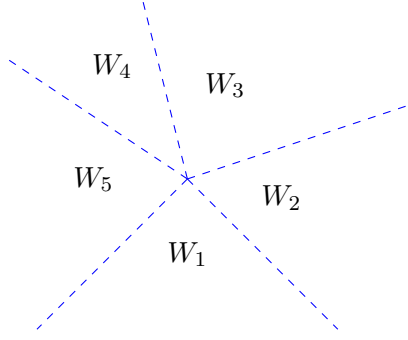
The product within each wedge reads
\bea
&&I\big|_{W_i}=\prod_{(m_2,m_3)\in W_i}\prod_{m_1\geq L^2_im_2+L_i^3m_3}\Big(\vec m\cdotp\vec \reeb+x\Big)\nn\\
&&\hspace{2cm}=\prod_{(m_2,m_3)\in W_i}\prod_{m_1\geq 0}\Big((\reeb^2+\reeb^1L_i^2)m_2+(\reeb^3+\reeb^1L_i^3)m_3+\reeb^1m_1+x\Big)~.\label{prod_I_W_i}\eea
We will denote by $\tilde \reeb_i$ the 2-vector
\bea
\tilde \reeb_i=(\reeb^2+\reeb^1 L_i^2,\reeb^3+ \reeb^1 L_i^3)~,\label{reeb_tilde}
\eea
which changes from one wedge to another.

The product over $m_1$ is now straightforward, and we have reduced the problem to the following. Consider two lines in $\BB{R}^2$ with rational slopes that bound $W_i$, how do we perform the summation (or the product, all the same) of the weight $\vec\xi\cdotp\vec n=\xi_1n_1+\xi_2n_2$
over the integral points between these two lines? We assume that the normals of the two lines are $v_1,\;v_2$, which are primitive integer 2-vectors, see Figure \ref{fig_add_lines}. Then we have the sum
\bea
\sum_{\vec n\cdotp v_1\geq0;~\vec n\cdotp v_2\leq0}\vec\xi\cdotp\vec n~.\nn
\eea
The strategy is to add more lines between the two given lines, so that the two normals of each pair of neighbouring lines form an $SL(2,\BB{Z})$ matrix, then one can, by an $SL(2,\BB{Z})$ matrix, transform the two lines into the $x$- and $y$-axis, in which situation the sum would be simple. Surely, one cannot know how many lines one would need to add, but so long as the process contains only finite number of steps, which we show next, the lack of explicitness need not hinder us.

Without loss of generality, one can assume $v_1=[0,1]$, i.e. the first line is the $x$-axis (by applying an $SL(2,\BB{Z})$ transformation, since $v_1$ is primitive). For definiteness, we also assume $v_2=[-p,q]$ with $\gcd(p,q)=1,~p,q>0$, the other possibilities can be treated entirely similarly, see Figure \ref{fig_add_lines}.
\begin{figure}[h]
\begin{center}
\begin{tikzpicture}[scale=.8]
\draw [->,blue] (0,0) -- (2,1.5) -- (1.25,2.5) node[above] {\small$v_2$};
\draw [->,blue] (0,0) -- (2.4,0) -- (2.4,1) node[right] {\small$v_1$};
\end{tikzpicture}
\hspace{2cm}
\begin{tikzpicture}[scale=.8]
\draw [blue] (0,0) -- (1,3) node[right] {\small{$[-23,17]$}};
\draw [red] (0,0) -- (1.4,2.1) node[right] {\small{$[-4,3]$}};
\draw [red] (0,0) -- (2,1) node[right] {\small{$[-1,1]$}};
\draw [blue] (0,0) -- (1.5,0) node[right] {\small{[0,1]}};
\end{tikzpicture}\caption{Sum between the two blue lines, depicted in the left panel. One can add more lines in between, as in the right panel. The numbers label the \emph{normal} of each line. The slopes of the lines are not drawn to scale.}\label{fig_add_lines}
\end{center}
\end{figure}
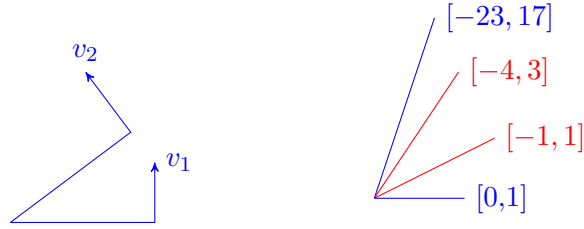
One simply observes that given two numbers $p,q>0$ coprime, one can find $s,t>0$ such that $pt-qs=1$ and that $p>s$, $q>t$. The proof is a simple exercise and is left for the reader,
otherwise consult \cite{Fulton}.
Then it is easy to see that $p/q>s/t$ so the new line has a smaller slope. Further $\det[-s,t;-p,q]=1$, which is part of we set out to achieve. One can continue this process, since the size $s,t$ as well as the slope decreases each time, the process will stop after finitely many steps. It is not at all important to know exactly about the lines added, so long as they exist.

We will now further subdivide each wedge of Figure \ref{fig_sub_div} using the algorithm described above, and get Figure \ref{fig_sub_div_II}. We denote by $\vec u_k$ the normals (counterclockwise pointing) of all the lines.
\begin{figure}[h]
\begin{center}
\begin{tikzpicture}[scale=1]

\draw [dashed,blue] (0,0) -- (-2,-2);
\draw [dashed,blue] (0,0) -- (2,-2) node[below] {\scriptsize{$i$}};;
\draw [dashed,blue] (0,0) -- (3,1);
\draw [dashed,blue] (0,0) -- (-.6,2.4);
\draw [dashed,blue] (0,0) -- (-2.4,1.6);

\node at (0,-1) {\small$W_1$};
\node at (1.25,-0.25) {\small$W_2$};
\node at (.5,1.25) {\small$W_3$};
\node at (-1,1.5) {\small$W_4$};
\node at (-1.25,0) {\small$W_5$};

\draw [dotted,blue] (0,0) -- (-1.4,-2) node[below] {\scriptsize{$k$}};
\draw [dotted,blue] (0,0) -- (-.8,-2) node[below] {\scriptsize{$k+1$}};
\node at (0,-2) {\small$\cdots$};
\draw [dotted,blue] (0,0) -- (.8,-2) node[below] {\scriptsize{$i-1$}};
\draw [dotted,blue] (0,0) -- (2.4,-1.2) node[below] {\scriptsize{$i+1$}};;
\draw [dotted,blue] (0,0) -- (2.4,0.2);
\node at (2.4,-0.2) {\small$\vdots$};
\draw [dotted,blue] (0,0) -- (2.8,1.5);
\node at (1.4,2) {\small$\ddots$};
\draw [dotted,blue] (0,0) -- (0,2.5);
\draw [dotted,blue] (0,0) -- (-2.4,1);
\node at (-2.4,-0.2) {\small$\vdots$};
\draw [dotted,blue] (0,0) -- (-2.4,-1.6);
\end{tikzpicture}
\caption{Further division of the $m_2$-$m_3$ plane, by adding lines. The normals of all lines are pointing counterclockwise.}\label{fig_sub_div_II}
\end{center}
\end{figure}
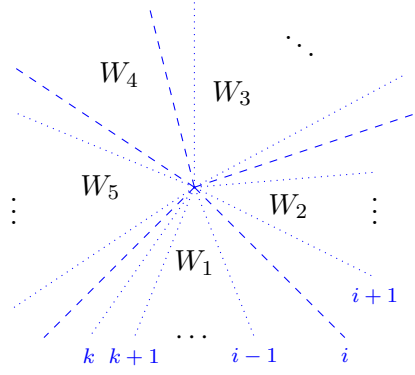
Now we have myriads of wedges over which we need to do the sum, as an example we consider first the product of (\ref{prod_I_W_i}) from (and including) the line $k$ up to (but excluding) the line $k+1$
\bea
 I\big|_{[k,k+1)}=\prod_{\vec n\cdot \vec u_k\geq0,\vec n\cdot \vec u_{k+1}<0}\prod_{m\geq 0}\Big(\tilde\reeb_1\cdotp \vec n+\reeb^1m+x\Big)~,\label{notation_ref}
 \eea
where $\tilde\reeb$ is defined in (\ref{reeb_tilde}). Since by assumption $\det[\vec u_k,\vec u_{k+1}]=1$, the product is simply
\bea I\big|_{[k,k+1)}&=&\prod_{\vec n\cdot \vec u_k\geq0,\vec n\cdot \vec u_{k+1}<0}\prod_{m\geq 0}\Big((\vec n\cdotp \vec u_k)(\tilde\reeb_1\times \vec u_{k+1})-(\vec n\cdotp \vec u_{k+1})(\tilde\reeb_1\times \vec u_k)+\reeb^1m+x\Big)\nn\\
&=&\prod_{m^{1,2,3}\geq0}\Big(m^2(\tilde\reeb_1\times \vec u_{k+1})+m^3(\tilde\reeb_1\times \vec u_k)+m^1\reeb^1+x+(\tilde\reeb_1\times \vec u_k)\Big)\nn\\
&=&\Gc_3\big(x+(\tilde\reeb_1\times \vec u_k)\big|\tilde\reeb_1\times \vec u_{k+1},\tilde\reeb_1\times \vec u_k,\reeb^1\big)^{-1}~,\nn\eea
where we use the short hand notation $\vec u\times\vec v=\det[u,v]$ for 2-vectors.

Before we do the product of the second factor in (\ref{sum_in_cone_II}), we need to make a technical remark. From the way all the dividing lines are chosen, one has
\bea
 \vec L_1+\vec u_i=\vec L_2~.\label{technical_detail}
 \eea
To see this, note the line $[y,z]$ separating $W_1$ and $W_2$ satisfies $[y,z]\cdotp (\vec L_1-\vec L_2)=0$, so its normal $\vec u_i$ is parallel to $\vec L_1-\vec L_2$. From the goodness of the cone, the 2-vector $\vec L_1-\vec L_2$ is primitive, thus $\vec u_i=\pm(\vec L_1-\vec L_2)$, and a little more thought would reveal the right sign. This relation holds for every line that separates two wedges $W_l$ and $W_{l+1}$.

The previous observation has the consequence that
\bea
 \tilde \reeb_1\times \vec u_i=\tilde \reeb_2\times \vec u_i~,\label{sum_in_cone_III}
 \eea
and thus it does not matter if one includes the contribution along the line $i$ in $W_1$ or $W_2$. Now for the second product (\ref{sum_in_cone_II}), we use this freedom to perform the product from (but excluding) the line $k$ up to (and including) the line $k+1$
\bea
II\big|_{(k,k+1]}&=&\prod_{\vec n\cdot \vec u_k>0,\vec n\cdot \vec u_{k+1}\leq0}\prod_{m\geq 0}\Big(\tilde\reeb_1\cdotp \vec n+\reeb^1m+\reeb^1-x\Big)\nn\\
&=&\prod_{m^{1,2,3}\geq0}\Big(m^2(\tilde\reeb_1\times \vec u_{k+1})+m^3(\tilde\reeb_1\times \vec u_k)+m^1\reeb^1+\reeb^1-x+(\tilde\reeb_1\times \vec u_{k+1})\Big)\nn\\
&=&\Gc_3\big(-x+\reeb^1+(\tilde\reeb_1\times \vec u_{k+1})\big|\tilde\reeb_1\times \vec u_{k+1},\tilde\reeb_1\times \vec u_k,\reeb^1\big)^{-1}~.\nn
\eea
Now one can combine $I\big|_{[k,k+1)}$ and $II\big|_{(k,k+1]}$, one gets the triple sine function
\bea I\big|_{[k,k+1)}\times II\big|_{(k,k+1]}=S_3\big(x+(\tilde\reeb_1\times \vec u_k)\big|\tilde\reeb_1\times \vec u_{k+1},\tilde\reeb_1\times \vec u_k,\reeb^1\big)~.\label{triple_sine_wedge}\eea

Note that in this way of dividing the cone, one always misses the points in (\ref{notation_ref}) with $\vec n=0,\,m\geq0$, but this can be done easily, and one gets
\bea
I_0=\prod_{m\geq 0}\big(\reeb^1m+x\big)\label{I_0}\eea
and for the factor $II$
\bea
 II_0=\prod_{m\geq 1}\big(\reeb^1m-x\big)~.\label{II_0}\eea
These two terms give
\bea
 I_0\cdotp II_0\sim \sin(\frac{\pi x}{\reeb^1})\sim e^{-\pi i\frac{x}{\sreeb^1}}(1-e^{2\pi i\frac{x}{\sreeb^1}})~.\label{extra_0}\eea
where $\sim$ means up to an overall multiplicative constant.   For the hypermultiplet, instead of (\ref{extra_0}), we will get
\bea 
e^{\pi i\frac{(x+im+\sreeb^1/2)}{\sreeb^1}}(1-e^{2\pi i\frac{(x+im+\sreeb^1/2)}{\sreeb^1}})^{-1}~.\label{extra_0_hyper}\eea
Later on, the second factors of (\ref{extra_0}) and (\ref{extra_0_hyper}) will cancel against terms coming from the $S_3$ function, while the first will be combined with the Bernoulli factors.

\subsection{Factorisation of the Triple Sines}

We will use the factorisation formula for the triple sine \cite{MR2101221}
\bea
&&  S_3 (z | \omega_1, \omega_2, \omega_3) = e^{-\frac{\pi i}{6} B_{3,3} (x|\omega_1, \omega_2, \omega_3)}
   (e^{2\pi i z/\go_2} ; e^{2\pi i \go_1/\go_2}, e^{2\pi i \go_3/\go_2})_\infty \nn \\
&&\hspace{2.5cm}\times (e^{2\pi i z/\go_1} ; e^{2\pi i \go_3/\go_1}, e^{2\pi i \go_2/\go_1})_\infty
   (e^{2\pi i z/\go_3} ; e^{2\pi i \go_1/\go_3}, e^{2\pi i \go_2/\go_3})_\infty~,\label{app-fact-33}\eea
and in what follows we will write $(x|y,z)$ instead of $(e^{2\pi ix};e^{2\pi iy},e^{2\pi iz})_{\infty}$.

Now the expression (\ref{triple_sine_wedge}) can be factorised
\bea
(\ref{triple_sine_wedge})&=&B\cdotp\big(\frac{x+\tilde\reeb_1\times \vec u_k}{\reeb^1}\big|\frac{\tilde\reeb_1\times \vec u_k}{\reeb^1},\frac{\tilde\reeb_1\times \vec u_{k+1}}{\reeb^1}\big)\nn\\
&&\hspace{1cm}
\big(\frac{x+\tilde\reeb_1\times\vec u_k}{\tilde\reeb_1\times \vec u_k}\big|\frac{\reeb^1}{\tilde\reeb_1\times\vec u_k},\frac{\tilde\reeb_1\times \vec u_{k+1}}{\tilde\reeb_1\times \vec u_k}\big)
\big(\frac{x+\tilde\reeb\times \vec u_k}{\tilde\reeb_1\times \vec u_{k+1}}\big|\frac{\reeb^1}{\tilde\reeb_1\times \vec u_{k+1}},\frac{\tilde\reeb_1\times \vec u_k}{\tilde\reeb_1\times \vec u_{k+1}}\big)\nn\\
&=&B\cdotp\big(\frac{x}{\reeb^1}\big|-\frac{\tilde\reeb_1\times \vec u_k}{\reeb^1},\frac{\tilde\reeb_1\times \vec u_{k+1}}{\reeb^1}\big)^{-1}\nn\\
&&
\big(\frac{x}{\tilde\reeb_1\times \vec u_k}\big|\frac{\reeb^1}{\tilde\reeb_1\times\vec u_k},\frac{\tilde\reeb_1\times \vec u_{k+1}}{\tilde\reeb_1\times \vec u_k}\big)
\big(\frac{x}{\tilde\reeb_1\times \vec u_{k+1}}\big|\frac{\reeb^1}{\tilde\reeb_1\times \vec u_{k+1}},-\frac{\tilde\reeb_1\times \vec u_k}{\tilde\reeb_1\times \vec u_{k+1}}\big)^{-1}\label{intermediate}~,\eea
where $B$ is the Bernoulli polynomial that we shall collect in subsection \ref{sec_BP} and we have also used (\ref{app-relations-1}).
 One can also use the factorisation in (\ref{factorisation_sine_sara}), then the Bernoulli polynomials do not occur.

The second factor of the first line of (\ref{intermediate}) can be simplified into
\bea \prod_k\big(\frac{x}{\reeb^1}\big|-\frac{\reeb^{\perp}\times \vec u_k}{\reeb^1},\frac{\reeb^{\perp}\times \vec u_{k+1}}{\reeb^1}\big)^{-1},\nn\eea
where $\reeb^{\perp}$ is the second and third component of $\vec\reeb$, i.e. $\reeb^{\perp}=\tilde\reeb_i-\reeb^1\vec L_i=[\reeb^2,\reeb^3]$. This manipulation is justified by using the periodicity of $(-|-,-)$. In appendix \ref{sec_Lemma} this product  is shown  to be
\bea \prod_k\big(\frac{x}{\reeb^1}\big|-\frac{\reeb^{\perp}\times \vec u_k}{\reeb^1},\frac{\reeb^{\perp}\times \vec u_{k+1}}{\reeb^1}\big)^{-1}=\big(1-\exp\big(\frac{2\pi ix}{\reeb^1}\big)\big)^{-1}.\nn\eea
This factor will cancel the second factor in (\ref{extra_0}) (or (\ref{extra_0_hyper}) in the case of hypermultiplet).

In the rest of this section, we focus on the second line of (\ref{intermediate}), which will give us a copy of the Nekrasov partition function for each corner of the moment map cone.
For every three neighbouring lines, say, $k-1$, $k$ and $k+1$ that are in the same wedge $W_1$, we will get the contribution
\bea
\big(\frac{x}{\tilde\reeb_1\times \vec u_k}\big|\frac{\reeb^1}{\tilde\reeb_1\times\vec u_k},\frac{\tilde\reeb_1\times \vec u_{k+1}}{\tilde\reeb_1\times \vec u_k}\big)
\big(\frac{x}{\tilde\reeb_1\times \vec u_k}\big|\frac{\reeb^1}{\tilde\reeb_1\times \vec u_k},-\frac{\tilde\reeb_1\times \vec u_{k-1}}{\tilde\reeb_1\times\vec u_k}\big)^{-1}~.\nn\eea
Here we make the observation that since $\vec u_{k-1}\times\vec u_k=\vec u_k\times\vec u_{k+1}=1$, one has
\bea
\vec u_{k-1}+\vec u_{k+1}=\BB{Z}\vec u_k~.\label{technical_detail_II}\eea
Consequently $\tilde \reeb_1\times\vec u_{k+1}+\tilde \reeb_1\times\vec u_{k-1}=\BB{Z}\tilde\reeb_1\times\vec u_k$, and the above combination cancels by using the periodicity of the special function $(-|-,-)$.

In contrast, take three lines as $i-1$, $i$ and $i+1$ with $i$ straddling two wedges $W_1$, $W_2$, then one gets instead the contribution
\bea\star=
\big(\frac{x}{\tilde\reeb_2\times \vec u_i}\big|\frac{\reeb^1}{\tilde\reeb_2\times\vec u_i},\frac{\tilde\reeb_2\times \vec u_{i+1}}{\tilde\reeb_2\times \vec u_i}\big)
\big(\frac{x}{\tilde\reeb_1\times \vec u_i}\big|\frac{\reeb^1}{\tilde\reeb_1\times \vec u_i},-\frac{\tilde\reeb_1\times \vec u_{i-1}}{\tilde\reeb_1\times \vec u_i}\big)^{-1}~.\nn\eea
One uses then (\ref{technical_detail_II}) and  (\ref{sum_in_cone_III}) to get
\bea\star=
\big(\frac{x}{\tilde\reeb_1\times \vec u_i}\big|\frac{\reeb^1}{\tilde\reeb_1\times\vec u_i},-\frac{\tilde\reeb_2\times \vec u_{i-1}}{\tilde\reeb_1\times \vec u_i}\big)
\big(\frac{x}{\tilde\reeb_1\times \vec u_i}\big|\frac{\reeb^1}{\tilde\reeb_1\times \vec u_i},-\frac{\tilde\reeb_1\times \vec u_{i-1}}{\tilde\reeb_1\times \vec u_i}\big)^{-1}~,\nn\eea
and that
\bea -\frac{\tilde\reeb_2\times \vec u_{i-1}}{\tilde\reeb_1\times \vec u_i}=-\frac{(\tilde\reeb_1+\reeb^1\vec u_i)\times \vec u_{i-1}}{\tilde\reeb_1\times \vec u_i}
=-\frac{\tilde\reeb_1\times \vec u_{i-1}-\reeb^1 }{\tilde\reeb_1\times \vec u_i}~.\label{technical_detail_IV}\eea
Now one invokes (\ref{modular_new}) and combine the two factors of $\star$
\bea \star=
\big(\frac{x}{\tilde\reeb_1\times \vec u_i}\big|-\frac{\tilde\reeb_2\times \vec u_{i-1}}{\tilde\reeb_1\times \vec u_i},\frac{\tilde\reeb_1\times \vec u_{i-1}}{\tilde\reeb_1\times \vec u_i}\big)=\big(\frac{x}{\tilde\reeb_1\times \vec u_i}\big|\frac{\tilde\reeb_2\times \vec u_{i+1}}{\tilde\reeb_1\times \vec u_i},\frac{\tilde\reeb_1\times \vec u_{i-1}}{\tilde\reeb_1\times \vec u_i}\big)~.\label{corner}\eea
To conclude, apart from the Bernoulli polynomials, the partition function receives a contribution of (\ref{corner}), for \emph{every corner of the moment map cone}.  If one were to use the factorisation (\ref{factorisation_sine_sara}), then the second factor there combines in a similar fashion into
\begin{align}
\star' &=
\big(-\frac{x}{\tilde\reeb_1\times \vec u_i}\big|\frac{\tilde\reeb_1\times \vec u_{i-1}-\reeb^1}{\tilde\reeb_1\times \vec u_i},-\frac{\tilde\reeb_1\times \vec u_{i-1}}{\tilde\reeb_1\times \vec u_i}\big) \nn \\
&=\big(\frac{\reeb^1-x}{\tilde\reeb_1\times \vec u_i}\big|-\frac{\tilde\reeb_1\times \vec u_{i-1}-\reeb^1}{\tilde\reeb_1\times \vec u_i},\frac{\tilde\reeb_1\times \vec u_{i-1}}{\tilde\reeb_1\times \vec u_i}\big)~.\label{corner_prime}
\end{align}
The same manipulation applies to the hypermultiplet, one needs only replace in the above formulae $x\to x+\reeb^1/2+im$.

Next we will show that this factor is the perturbative Nekrasov partition function on $S^1\times\BB{C}^2$.
Since the wedges $W$ correspond to the faces of the moment map cone, one observes that if the normals to face 1 and 2 are $\vec v$ and $\vec v'$, i.e. $\vec v=[1,-\vec L_1]$ and $\vec v'=[1,-\vec L_2]$, then
\bea
\det[\vec v,\vec v',\vec\reeb]=\det\left(
                                          \begin{array}{ccc}
                                            1 & 1 & \reeb^1 \\
                                            -L_1^2 & -L_2^2 & \reeb^2 \\
                                            -L_1^3 & -L_2^2 & \reeb^3 \\
                                          \end{array}\right)=\tilde\reeb_1\times\vec u_i~.\nn\eea
Thus one recognizes the quantity $\tilde\reeb_1\times\vec u_i$ as the inverse circumference $2\pi/\gb$ of the Reeb orbit above the corner at the intersection of face 1 and 2
 (see Figure \ref{fig_toric_cone}).

For the equivariant parameters, let $\vec n=[0,-\vec u_{i+1}]$, one observes that
\bea \det[\vec v,\vec v',\vec n]=\det\left(
                                                  \begin{array}{ccc}
                                                    1 & 1 & 0 \\
                                                    -\vec L_1 & -\vec L_2 & -\vec u_{i+1} \\ \end{array}\right)=
\det\left(
                                                  \begin{array}{ccc}
                                                    0 & 1 & 0 \\
                                                    \vec u_i & -\vec L_2 & -\vec u_{i+1} \\ \end{array}\right)=1~.\nn\eea
Then from the recipe (\ref{recipe_bee}) for $\ep,\;\ep'$, one gets
\bea
&&\ep=\det[\vec n,\vec \reeb,\vec v']=\tilde\reeb_2\times\vec u_{i+1}~,\nn\\
&& \ep'=\det[\vec v,\vec \reeb,\vec n]=-\tilde\reeb_1\times\vec u_{i+1}=\tilde\reeb_1\times\vec u_{i-1}+\BB{Z}\tilde\reeb_1\times\vec u_i~.\nn\eea
From this we see that the partition function receives one copy of the perturbative Nekrasov partition function for each corner of the toric moment cone, or for each closed Reeb orbit, with the expected equivariant parameters
\bea
 \star=\big(\frac{\gb}{2\pi} x\big|\frac{\gb}{2\pi}\ep,\frac{\gb}{2\pi}\ep'\big)~,~~~\star'=\big(\frac{\gb}{2\pi}(\reeb^1-x)\big|\frac{\gb}{2\pi}\ep,\frac{\gb}{2\pi}\ep'\big)~.\nn\eea

If we adopt the second factorisation of the triple sine (\ref{factorisation_sine_sara}), we will get the following
\bea
Z^{pert}=
\int\limits_{\FR{t}}
\resizebox{0.78\hsize}{!} {$
da~e^{-\frac{8\pi^3 r}{\gYM^2}\varrho\,\Tr[a^2]}\cdotp\frac{\prod\limits_{i=1}^{\tt n}\big({\det}_{adj}'\big(i\frac{\gb_i}{2\pi}a\big|\frac{\gb_i}{2\pi}\ep_i,\frac{\gb_i}{2\pi}\ep_i'\big)\big(a\to -i\reeb^1-a\big)\big)^{1/2}}{\prod\limits_{i=1}^{\tt n}\big(\det_{\underline{R}}\big(i\frac{\gb_i}{2\pi}(a+m-i\reeb^1/2)\big|\frac{\gb_i}{2\pi}\ep_i,\frac{\gb_i}{2\pi}\ep_i'\big)\big(a+m\to -a-m\big)\big)^{1/2}}
$}\label{Z_pert_fac}
\eea
where the index $i$ runs over all the $\tt n$ closed Reeb orbits. This way of writing the factorization, though involving a square root, is manifestly symmetric under $\udl{R}\to\udl{\bar R}$.

\subsection{Collection of the Bernoulli Polynomials}\label{sec_BP}

In this section we collect the Bernoulli polynomials left over from (\ref{intermediate}). The Bernoulli polynomial $B_{3,3}$ is defined in (\ref{bernoulli}).
From the contribution from line $k$ to line $k+1$, one receives
\bea
-\frac{\pi i}{6}B_{3,3}\big(x+(\tilde\reeb_1\times \vec u_k)\big|\tilde\reeb_1\times \vec u_{k+1},\tilde\reeb_1\times \vec u_k,\reeb^1\big)
=\frac{\pi i}{6}B_{3,3}\big(x\big|\tilde\reeb_1\times \vec u_{k+1},-\tilde\reeb_1\times \vec u_k,\reeb^1\big)~,\nn\eea
where (\ref{bernoulli_prop}) is used.

We collect the $x^3$ term first
\bea
 \textrm{coef of } x^3=\frac{\pi i}{6}\frac{1}{\reeb^1(\tilde\reeb_1\times\vec u_{k+1})(-\tilde\reeb_1\times \vec u_k)}~.\nn\eea
The right hand side is actually proportional to the area of the part of a face (face 1 in this particular instance, see Figure \ref{fig_toric_cone}) bounded by the three planes $\vec y\cdotp[0,\vec u_k]=0$, $\vec y\cdotp[0,\vec u_{k+1}]=0$ and $\vec y\cdotp\vec\reeb=1/2$, $\vec y\in\BB{R}^3$, see Figure \ref{fig_erea_triangle}.
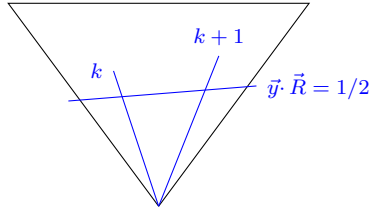
\begin{figure}[h]
\begin{center}
\begin{tikzpicture}[scale=1]
\draw (0,0) -- (2,2.7) -- (-2,2.7) -- (0,0);
\draw [-,blue] (0,0) -- (.8,2) node[above] {\scriptsize{$k+1$}};
\draw [-,blue] (0,0) -- (-.6,1.8) node[left] {\scriptsize{$k$}};
\draw [-,blue] (-1.2,1.4) -- (1.3,1.6) node[right] {\scriptsize{$\vec y\cdotp \vec\reeb=1/2$}};
\end{tikzpicture}
\caption{The big triangle is the face 1, and we are interested in the area enclosed on face 1 by three planes: $\vec y\cdotp[0,\vec u_k]=0$, $\vec y\cdotp[0,\vec u_{k+1}]=0$ and $\vec y\cdotp\vec\reeb=1/2$.}\label{fig_erea_triangle}
\end{center}
\end{figure}
Indeed, the area is given by the expression
\bea
A=\frac18 \frac{|w_2|\det[w_1,w_3,w_2]}{\det[w_1,w_2,\vec\reeb]\cdotp\det[w_3,w_2,\vec\reeb]}~,~~~\textrm{where}~~w_1=[0,\vec u_k]~,~w_2=[1,-\vec L_1]~,~w_3=[0,\vec u_{k+1}]~.\nn\eea
Working this out, we have
\bea A=\frac18 \frac{|[1,-\vec L_1]|}{(\tilde\reeb_1\times u_k)(\tilde\reeb_1\times\vec u_{k+1})}~.\nn\eea
Coming back to the coefficient of $x^3$, summing over $k$ we get
\bea
 \textrm{coef of } x^3=-\frac{4\pi i}{3\reeb^1}\sum_{i=1}^{\tt n}\frac{1}{|\vec v_i|}A_i~,\nn
 \eea
where $A_i$ is the area of face $i$ topped off by the plane $\vec y\cdotp \vec\reeb=1/2$, and $i$ runs over all faces.

We collect the $x^2$ term next
\bea
\textrm{coef of } x^2&=&-\frac{\pi i}{4}\frac{\reeb^1+\tilde\reeb_1\times(\vec u_{k+1}-\vec u_k)}{\reeb^1(\tilde\reeb_1\times\vec u_{k+1})(-\tilde\reeb_1\times \vec u_k)}\nn\\
&=&\frac{\pi i}{4}\Big(\frac{1}{(\tilde\reeb_1\times\vec u_{k+1})(\tilde\reeb_1\times \vec u_k)}+\frac{1}{\reeb^1(\tilde\reeb_1\times \vec u_k)}
-\frac{1}{\reeb^1(\tilde\reeb_1\times\vec u_{k+1})}\Big)~.\nn\eea
The last two terms will drop once we sum over all $k$ (using again (\ref{sum_in_cone_III})).
and summing over all $k$, one gets
\bea
\textrm{coef of } x^2=2\pi i\sum_{i=1}^{\tt n}\,\frac{1}{|v_i|}A_i~.\nn\eea

For the $x^1$ term, we get
\bea
 \textrm{coef of }x^1=\frac{\pi i}{12}\Big(\frac{\go_1}{\go_2\go_3}+\frac{1}{\go_1}\big(\frac{\go_2}{\go_3}+\frac{\go_3}{\go_2}+3\big)+3(\frac{1}{\go_3}+\frac{1}{\go_2})\big)~,\nn\eea
where $\go_1=\reeb^1$, $\go_2=\tilde\reeb_1\times \vec u_{k+1}$ and  $\go_3=-\tilde\reeb_1\times \vec u_k$. Taking the sum over $k$, the last term will drop, and the first term has been dealt with above. For the middle term, we only need to investigate the following
\bea \sum_k\Big(-3+\frac{\tilde\reeb_1\times \vec u_{k+1}}{\tilde\reeb_1\times \vec u_k}+\frac{\tilde\reeb_1\times \vec u_k}{\tilde\reeb_1\times \vec u_{k+1}}\Big)
~.\nn\eea
Using (\ref{technical_detail_II}),
\bea \frac{\tilde\reeb_1\times (\vec u_{k-1}+\vec u_{k+1})}{\tilde\reeb_1\times \vec u_k}\in \BB{Z}\nn\eea
if $k-1$, $k$ and $k+1$ are in the same wedge. Otherwise, if $k$ separates $W_1$ and $W_2$ one gets
\bea \frac{\tilde\reeb_2\times\vec u_{k+1}}{\tilde\reeb_2\times \vec u_k}+\frac{\tilde\reeb_1\times \vec u_{k-1}}{\tilde\reeb_1\times \vec u_k}=
\frac{-\tilde\reeb_2\times\vec u_{k-1}+\tilde\reeb_1\times \vec u_{k-1}}{\tilde\reeb_1\times \vec u_k}
+\BB{Z}=\frac{\reeb^1}{\det[\vec\reeb,\vec v_1,\vec v_2]}+\BB{Z}~,\nn\eea
where (\ref{technical_detail_IV}) is used.
And one recognize the last combination as proportional to the circumference of the closed Reeb orbits at the corner of the intersection of faces $1$ and $2$. In total the $x^1$ term gives
\bea \textrm{coef of }x^1=\frac{\pi i}{12}\Big(-8\reeb^1\sum_{i=1}^{\tt n}\frac{1}{|\vec v_i|}A_i-\frac{1}{2\pi}\sum_{i=1}^{\tt n}\gb_i-\frac{c}{\reeb^1}\Big)~,\nn\eea
where the undetermined integer is named $c$ and it will be shown to be $-12$ at the end of this section.

Finally we come to the $x^0$ term. One might wonder why do we bother with this since it is just a constant and we have been discarding constants all along, but the point is that the same type of terms appearing here will appear in the asymptotic behaviour of the partition function where they will be important. The $B_{3,3}$ has the constant term
\bea
\textrm{coef of }x^0=-\frac{\pi i}{24}\Big(3+\frac{\go_1}{\go_2}+\frac{\go_2}{\go_1}+\frac{\go_2}{\go_3}+\frac{\go_3}{\go_2}+\frac{\go_3}{\go_1}+\frac{\go_1}{\go_3}\Big)~,\nn\eea
with the same $\go$'s as above.
Taking the sum over $k$ one gets
\bea \frac{\pi i}{24}\Big(\frac{\reeb^1}{2\pi}\sum_{i=1}^{\tt n}\gb_i+c\Big)~.\nn\eea

To summarise the collection of Bernoulli polynomials gives
\bea
\pi i\big(-\frac{4x^3}{3\reeb^1}+2x^2-\frac{2}{3}\reeb^1x\big)\sum_{i=1}^{\tt n}\frac{1}{|\vec v_i|}A_i+\pi i\big(-\frac{1}{12}x+\frac{1}{24}\reeb^1\big)\frac{1}{2\pi}\sum_{i=1}^{\tt n}\gb_i+\pi ic\big(-\frac{x}{12\reeb^1}+\frac{1}{24}\big)~.\nn\eea
As an aside when the cone corresponds to a CY toric manifold, which is the case we are dealing with, one can write the sum of volume of faces above as the volume of the manifold $X$.
One uses the fact that the end points of the normals $\vec v_i$ lie on a hyperplane, the sum of the volume of the faces above can be written as
\bea
 \sum_{i=1}^{\tt n}\,\frac{1}{|v_i|}A_i=6\reeb^1\textrm{vol}_{\Gd_{\sreeb}^{1/2}}=\frac{\reeb^1}{(2\pi)^3}\textrm{vol}_X~,\label{volume}\eea
where $\Gd_{\sreeb}^{1/2}$ is the intersection $C_{\mu}(X)\cap\{\vec r\in\BB{R}^3|\vec r\cdotp\vec\reeb\leq1/2\}$. The above relation is derived in \cite{Martelli:2005tp}, it was also shown in that paper that
\bea
\int\limits_{C(X)^1}R_{C(X)}=(2\reeb^1-6)\textrm{vol}_X~,\label{ricci_scalar}
\eea
where $R_{C(X)}$ is the Ricci scalar of the metric cone $C(X)$, and $C(X)^1$ is the metric cone cut off at $r\leq1$, see section \ref{sec_TSEM} for notations.

To apply this result to the vector multiplet, one can discard the odd powers of $x$, since $x=i\bra a,\gl\ket$, $a\in \FR{t}$ and $\gl$ runs over all the roots, so the odd powers of $x$ cancel out. We get
\bea
B_{vec}(x)=2\pi ix^2\sum_{i=1}^{\tt n}\frac{1}{|\vec v_i|}A_i+\frac{\pi i}{24}\frac{\reeb^1}{2\pi}\sum_{i=1}^{\tt n}\gb_i-\frac{i\pi}{2}~.\label{Bvec}
\eea
For a hypermultiplet with mass $m$,  one needs to remember the contribution from the first factor of (\ref{extra_0_hyper}), and one gets
\bea
B_{hyp}(x)=\frac{4\pi i}{3}\big(\frac{1}{\reeb^1}(x+im)^3-\frac{1}{4}\reeb^1(x+im)\big)\sum_{i=1}^{\tt n}\frac{1}{|\vec v_i|}A_i+\frac{\pi i}{12}\big(x+im\big)\frac{1}{2\pi}\sum_{i=1}^{\tt n}\gb_i+\frac{\pi i}{2}.\label{Bhyp}\eea

We will now prove $c=-12$ (see page 44 \cite{Fulton}). First, one needs to establish that given a subdivision of the plane, the number $c$ is unchanged if one inserts further lines. To see this, let $\vec v_{i-2}$ $\vec v_{i-1}$, $\vec v_i$ and $\vec v_{i+1}$ be the normals to four consecutive lines such that $\vec v_k\times \vec v_{k+1}=1,~k=i-2,\cdots,i$, and we can assume that $\vec v_{i-1}=[-1,0]$ and $\vec v_i=[0,1]$. We insert a fifth line between $i-1$ and $i$, with normal $\vec u$, then one must have $\vec v_{i-1}+\vec v_i=\vec u=[-1,1]$. Doing this would change $c$ by
\bea
\gd c=\frac{\tilde\reeb\times(\vec u-\vec v_{i-1})}{\tilde\reeb\times \vec v_i}-3+\frac{\tilde\reeb\times(\vec v_{i+1}+\vec v_{i-1})}{\tilde\reeb\times \vec u}
+\frac{\tilde\reeb\times(\vec u-\vec v_i)}{\tilde\reeb\times \vec v_{i-1}}=1-3+1+1=0~.\nn\eea
One can go further and establish that $c$ does not change if we add $k$ redundant lines in between $i-1$ and $i$. To see this, if one of the $k$ lines we add has normal $\vec u=[-1,1]$, then since $\vec u\times \vec v_{i-1}=\vec v_i\times \vec u=1$, and there are fewer lines between either $\vec u$, $\vec v_{i-1}$ or $\vec u$, $\vec v_i$, and the proof follows from an induction. Next we show that such a line can always be found among the $k$ lines. Assume first that all $k$ lines are between $\vec u$ and $\vec v_{i-1}$ (resp. $\vec v_i$), then the last (resp. first) of these lines must have normal $\vec u$, and we are finished. In the remaining case, that is, there are lines between $\vec u$, $\vec v_{i-1}$ as well as between $\vec u$, $\vec v_i$. Assume that none of the $k$ lines have normal $\vec u$, then the two lines right next to it must have primitive normals $[-a,b]$ and $[-c,d]$ with $a,b,c,d>0$ and $a>b\geq1$, $d>c\geq 1$, then $\det[-c,d;-a,b]=-bc+ad>1$ and we get a contradiction.

It is also easy to check that for the case of three standard lines with normals $[0,1]$, $[-1,0]$ and $[1,-1]$, then $c=-12$.
With this understanding, now given any subdivision problem consisting of a set of $\tt n$ lines with normals $\vec v_i$, $i=1,\cdots\tt n$, one add to the list three standard lines with the above normals. It does not matter if one of the original lines happen to coincide with the standard lines, but for definiteness, let us assume otherwise. Now one can follow the subdivision algorithm to add more lines to the list of ${\tt n}+3$ lines. This subdivision certainly solves the subdivision problem of the original list of $\tt n$ lines, but it also
can be viewed as adding redundant lines to the set of three standard lines, and hence $c=-12$ from the above argument.

\section{The Asymptotic Behaviour and Large $N$}
Using the method of subdividing the moment map cone, we can now give a general formula for the asymptotic behaviour, expressed in terms of the geometrical data from the moment map cone.

In the two products of (\ref{sum_in_cone_I}) and  (\ref{sum_in_cone_II}), we give $x$ a small real part and send its imaginary part to infinity. As usual, the infinite product is taken under the zeta-function regularisation
\bea
 \log I=-\frac{\partial}{\partial s}\frac{1}{\Gc(s)}\int\limits_0^{\infty}\sum_{\vec m\in C_{\mu}(X)\cap\BB{Z}^3}e^{-(\vec m\cdotp\vec \sreeb+x)t}t^{s-1}\,dt\Big|_{s=0}~,\nn
 \eea
and $\log II$ is obtained by replacing $x=\reeb^1-x$.
The summation will now be done as in the earlier sections by dividing $C_{\mu}(X)$. In the $i^{th}$ wedge between line $k$ and $k+1$, one gets (see (\ref{notation_ref}) and Figure \ref{fig_sub_div_II} for the explanation of the notation)
\bea \log I\big|_{[k,k+1)}=-\frac{\partial}{\partial s}\frac{1}{\Gc(s)}\int\limits_0^{\infty}\frac{e^{-(x+\tilde\sreeb_1\times \vec u_k)t}}{(1-e^{-\tilde\sreeb_1\times \vec u_{k+1}t})(1-e^{-\tilde\sreeb_1\times \vec u_kt})(1-e^{-\sreeb^1t})}t^{s-1}\,dt\Big|_{s=0}~,\nn\\
\log II\big|_{(k,k+1]}=-\frac{\partial}{\partial s}\frac{1}{\Gc(s)}\int\limits_0^{\infty}\frac{e^{-(\sreeb^1-x+\tilde\sreeb_1\times \vec u_{k+1})t}}{(1-e^{-\tilde\sreeb_1\times \vec u_{k+1}t})(1-e^{-\tilde\sreeb_1\times \vec u_kt})(1-e^{-\sreeb^1t})}t^{s-1}\,dt\Big|_{s=0}~.\nn\eea
The large $\im x$ behaviour is then given by taking the Laurent series of the denominator at $t=0$ up to $t^0$ and then performing the integral. The details can be found in section 6 of \cite{Qiu:2013pta}, here we just give the result
\bea -\log I\big|_{[k,k+1)}-\log II\big|_{(k,k+1]}=\frac{i\pi}{6}\textrm{sgn}(\im x)B_{3,3}(x|\go_1,-\go_2,\go_3)
~.\nn\eea
where $\go_1=\reeb^1$, $\go_2=\tilde\reeb_1\times \vec u_k$ and $\go_3=\tilde\reeb_1\times \vec u_{k+1}$.

To apply this result to the vector multiplet, one can discard the even powers of $x$, since $x=i\bra a,\gl\ket$, $a\in \FR{t}$ and we shall be summing over all the roots $\gl$, so the even powers of $x$ cancel out. We are left with
\bea
&&-(\log I\big|_{[k,k+1)}+\log II\big|_{(k,k+1]})\big|_{vec}  \nn \\
& &=\frac{i\pi\textrm{sgn}(\im x)}{12\go_1\go_2\go_3}
\Big(2x^3+x(\go_1^2+\go_2^2+\go_3^2-3\go_1\go_2-3\go_2\go_3+3\go_3\go_1)\Big)~.\nn\eea
The assemblage of these contributions from all wedges is entirely similar to the treatment of the Bernoulli polynomials in subsection \ref{sec_BP}, we get
\bea -(\log I+\log II)\big|_{vec}=i\pi\textrm{sgn}(\im x)\Big(\big(\frac{x^3}{3\reeb^1}+\frac{\reeb^1x}{6}\big)\sum_i\frac{4}{|v_i|}A_i
+\frac{x}{12}\big(\frac{1}{2\pi}\sum_i\gb_i+\frac{c}{\reeb^1}\big)\Big)~.\nn\eea
The integer $c=-12$ was introduced in the previous section. One must not forget the contribution from the factors of (\ref{extra_0}), which gives
\bea -\frac{i\pi}{\reeb^1}\textrm{sgn}(\im x)x~.\label{extra_00}\eea
and the total asymptotic behaviour from the vector multiplet is
\bea V^{asymp}_v(x)=-i\pi\textrm{sgn}(\im x)\Big(\big(\frac{x^3}{3\reeb^1}+\frac{\reeb^1x}{6}\big)\sum_i\frac{4}{|v_i|}A_i
+\frac{x}{12}\frac{1}{2\pi}\sum_i\gb_i\Big)~.\label{asymp_v}\eea

For the hypermultiplet $x=\bra \gs,\mu\ket$, but the weights of a general representation may not be symmetric. Also remembering the shift $x\to x+\reeb^1/2$, one gets
\bea &&-(\log I\big|_{[k,k+1)}+\log II\big|_{(k,k+1]})\big|_{hyp}=\frac{i\pi\textrm{sgn}(\im x)}{72\go_1\go_2\go_3}
\Big(12x^3+18x^2(\go_2-\go_3)\nn\\
&&\hspace{4cm}-3x(\go_1^2-2\go_2^2-2\go_3^2+6\go_2\go_3)-3\go_2\go_3(\go_2-\go_3)-\frac32\go_1^2(\go_2-\go_3)\Big)~.\nn\eea
Now as we assemble the contributions from all wedges, the even powers of $x$ drop again
\bea -(\log I+\log II)\big|_{hyp}=i\pi\textrm{sgn}(\im x)\Big(\big(\frac{x^3}{3\reeb^1}-\frac{\reeb^1x}{12}\big)\sum_i\frac{4}{|v_i|}A_i
+\frac{x}{12}\big(\frac{1}{2\pi}\sum_i\gb_i+\frac{c}{\reeb^1}\big)\Big)~,\nn\eea
The factors of (\ref{extra_0}) gives a similar contribution as in (\ref{extra_00}),
and in total the asymptotic behaviour from the hypermultiplet is
\bea
 V^{asymp}_h(x)=i\pi\textrm{sgn}(\im x)\Big(\big(\frac{x^3}{3\reeb^1}-\frac{\reeb^1x}{12}\big)\sum_i\frac{4}{|v_i|}A_i
+\frac{x}{12}\frac{1}{2\pi}\sum_i\gb_i\Big)~.\label{asymp_h}\eea
To summarise, asymptotically, the matrix model integral is given by
\bea
Z^{pert}\sim \int\limits_{\FR{t}}da~e^{-\frac{8\pi^3 r}{\gYM^2}\varrho\,\Tr[a^2]}\cdotp
e^{\Tr_{adj}V^{aymp}_v(ia)}\cdotp e^{\Tr_{\underline{R}}V^{asymp}_h(ia)}~,\label{Z_asymp}\eea
with $V^{aymp}_{v,h}$ given in (\ref{asymp_v}) and  (\ref{asymp_h}). This seems a better way of presenting the asymptotic behaviour of the potential than the way  it was done in
\cite{Qiu:2013pta}, since the role played by the geometry is more transparent now.

Using these asymptotics and following the analysis  from  \cite{Kallen:2012zn} we  get the free energy at the large $N$ limit for the vector multiplet coupled to
a hypermultiplet in adjoint with mass $m$
\bea
   F = - \log Z  = - \frac{\gYM^2 N^3}{96 \pi r} \varrho \Big ( \frac{1}{4} (\reeb^1)^2 + m^2 \Big )^2~.\nn
\eea
for a squashed toric SE manifold. To go to the SE metric, one only needs to set $\reeb^1=3$ \cite{Martelli:2005tp}.
The result is identical to that of the theory on $S^5$ up to a volume factor $\varrho$ as expected.

\section{Summary}\label{s-summary}

In this paper we have derived the full perturbative partition function for the SYM coupled to hypermultiplets on any 5D toric simply connected SE manifold $X$.
 We have calculated the equivariant answer which keeps track of three $U(1)$ isometries on $X$. The actual 5D calculation can be reduced to the counting of
  holomorphic functions on the corresponding CY cone $C(X)$. Thus it is very natural to ask if there is anything deep in this relation to 6D counting besides being a mere technical
   trick. It will be extremely interesting to construct an intrinsically 6D theory which will do the same counting. Another natural question is if the contact instantons
    (localisation locus for 5D theory) has a natural lift to 6D. Somehow it is conceivable that the counting of contact instantons on $X$ also
     reduces to some counting problems on $C(X)$.

 Another important result of this paper is the factorization property of the full perturbative answer on $X$ into copies of perturbative
 Nekrasov partition functions on $\mathbb{C}^2 \times S^1$, with the twisting parameters controlled by the toric data of $X$. It is natural to conjecture
  that the full partition function on $X$ is given by gluing the copies of full Nekrasov partition function with the same set of twisting data as in the perturbative
   sector, however a constructive proof of this conjecture from the first principle is beyond us so far.

 A puzzle that we do not resolve is the following.
 While proving the factorisation we have studied the special function $S^X_3$ depending on $X$ through its toric data. When $X$ is simply connected, the zero instanton localisation locus  consists of just the zero connection, and the answer is given in terms of $S_3^X$. When $X$ is not simply connected, one does not a priori have a Killing spinor.  Moreover we
  would need to take into account all non-trivial flat connections to produce the complete perturbative partition function. From physical considerations, one expects that the contribution of all the flat connections together should factorise, but not individually. However, our proof of the factorisability of $S^X_3$ does not require the simply connectedness. One possible explanation is that the contribution from the zero connection is special and factorises all by itself.
 It would be extremely interesting to investigate the localisation for non-simply connected manifolds and the corresponding factorisation properties.

\bigskip
\bigskip
\bigskip
{\bf Acknowledgements} We thank Chris Herzog and Sara Pasquetti    for discussions.
  We are grateful to Charles Boyer for e-mail exchange regarding SE geometry.
 M.Z. thanks KITP, Santa Barbara for the hospitality where this project has been initiated.
  The research of J.Q. is supported by the Luxembourg FNR grant PDR 2011-2, and by the UL grant GeoAlgPhys 2011-2013. The research of M.Z. is supported in part by Vetenskapsr\r{a}det under grant $\sharp$ 2011-5079 and  in part by the National Science Foundation under Grant No. NSF PHY11-25915.

\appendix
\section{Special Functions}\label{a-special}

\subsection{Definitions of special functions}

The special function $(x|a_1,\cdots,a_n)_{\infty}$ was introduced in \cite{MR2101221}.
 It is defined differently in different domains
\bea &&(x|a_1,\cdots,a_n)_{\infty}=\prod_{i_1,\cdots,i_n\geq0}\big(1-xa_1^{i_1}\cdots a_{k-1}^{i_{k-1}}\,a_k^{-(i_k+1)}\cdots a_n^{-(i_n+1)}\big)^{-(-1)^{n-k}}~,\label{special_Narukawa}\\
&&\hspace{7cm} |a_1|<1,\cdots,|a_{k-1}|<1,~|a_k|>1,\cdots,|a_n|>1~.\nn\eea
This function is symmetric in the $n$ arguments $a_i$, but it is not defined if any $|a_i|=1$. These functions enjoy the property
\bea
 (x|a_1, \cdots , a_r)_{\infty} = \frac{1}{ (a_{j}^{-1}x| a_1, \cdots, a_j^{-1}, \cdots , a_r)_{\infty} }~.\label{app-relations-1}\eea

Often we will use the short hand
\bea
(e^{2\pi iz}|e^{2\pi i\go_1},\cdots,e^{2\pi i\go_n})_{\infty}=(z|\go_1,\cdots,\go_n)~.\label{short_hand}\eea
One needs to remember that when using the latter notation, the function is periodic under shift by an integer of any of the arguments.

\begin{lemma}\label{lem_special_fun}
  \bea
  \frac{(x|a,b)_{\infty}}{(x|a,ab)_{\infty}}=(x|b^{-1},ab)_{\infty}^{-1}~.\label{modular_new}\eea
\end{lemma}
\begin{proof}
  We prove the lemma case by case, first let $|a|<1$ and $|b|<1$, then
  \bea
   \frac{(x|a,b)_{\infty}}{(x|a,ab)_{\infty}}&=&\frac{\prod\limits_{i,j\geq0}(1-xa^ib^j)}{\prod\limits_{i,j\geq0}(1-xa^i(ab)^j)}
 =\prod\limits_{i\geq 0,\; j>i}(1-xa^ib^j)\nn \\
 &=&\prod\limits_{i,j\geq 0}(1-xb(ab)^ib^j)=(xb|ab,b)_{\infty}=(x|ab,b^{-1})_{\infty}^{-1}~.\nn\eea
  If instead $|a|<1$, $|b|>1$ but $|ab|<1$, then
    \bea
     \frac{(x|a,b)_{\infty}}{(x|a,ab)_{\infty}}&=&\frac{1}{\prod\limits_{i,j\geq0}(1-xa^ib^{-j-1})\prod\limits_{i,j\geq0}(1-xa^i(ab)^j)}\nn\\
  &=&\frac{1}{\prod\limits_{i\geq 0,\; j\leq i}(1-xa^ib^j)}=\frac{1}{\prod\limits_{i,j\geq 0}(1-xb^{-i}(ab)^j)}=(xb|b,ab)_{\infty}=(x|b^{-1},ab)_{\infty}^{-1}.\nn\eea
  But if $|ab|>1$
    \bea
  \frac{(x|a,b)_{\infty}}{(x|a,ab)_{\infty}}&=&\frac{\prod\limits_{i,j\geq0}(1-xa^i(ab)^{-j-1})}{\prod\limits_{i,j\geq0}(1-xa^ib^{-j-1})} =\prod_{j\geq0,-j-1\leq i<0}(1-xa^ib^{-j-1}) \nn \\
  &=&\prod_{k,l\geq0}(1-xa^{-k-1}b^{k+l+1})=(xb|ab,b)_{\infty}=(x|ab,b^{-1})_{\infty}^{-1}.\nn\eea
  By switching the role of $a,b,ab$ one can obtain the other cases\qed
\end{proof}

We will also make use of the multiple Gamma function, defined as a $\zeta$-regulated product
\bea
\Gamma_r=\prod_{n_1,\cdots,n_r=0}^{\infty}\big(n_1\go_1+\cdots+ n_r\go_r+x)^{-1}~,\label{mult_gamma}\eea
the domain of definition is that all $\go_i\in\BB{C}$ should lie on the same side of some straight
line through the origin and $x\in\BB{C}$.

The multiple sine function is defined as
\bea
 S_r(x|\go_1,\cdots,\go_r)=\Gc_r(x|\go_1,\cdots,\go_r)^{-1}\Gc_r(\sum_{i=1}^r\go_i-x|\go_1,\cdots,\go_r)^{(-1)^r}~.\label{mult_sine}\eea
The multiple sine function has an important factorisation property, see property 5 in \cite{MR2101221}, we shall only give the the case $r=3$
\bea
&& S_3(x|\go_1,\cdots,\go_r)=e^{-\frac{\pi i}{6}B_{3,3}(x|\go_1,\cdots,\go_3)}\nn\\
&&\hspace{0.5 cm}\big(e^{2\pi i\frac{x}{\go_1}}\big|e^{2\pi i\frac{\go_2}{\go_1}},e^{2\pi i\frac{\go_3}{\go_1}}\big)_{\infty}
\big(e^{2\pi i\frac{x}{\go_2}}\big|e^{2\pi i\frac{\go_1}{\go_2}},e^{2\pi i\frac{\go_3}{\go_2}}\big)_{\infty}
\big(e^{2\pi i\frac{x}{\go_3}}\big|e^{2\pi i\frac{\go_1}{\go_3}},e^{2\pi i\frac{\go_2}{\go_3}}\big)_{\infty}~,\label{factorisation_sine}\eea
or one may have the factorisation
\bea
 && S_3(x|\go_1,\cdots,\go_r)=e^{\frac{\pi i}{6}B_{3,3}(x|\go_1,\cdots,\go_3)}\nn\\
&&\hspace{1.5cm}\big(e^{-2\pi i\frac{x}{\go_1}}\big|e^{-2\pi i\frac{\go_2}{\go_1}},e^{-2\pi i\frac{\go_3}{\go_1}}\big)_{\infty}
\big(e^{-2\pi i\frac{x}{\go_2}}\big|e^{-2\pi i\frac{\go_1}{\go_2}},e^{-2\pi i\frac{\go_3}{\go_2}}\big)_{\infty}
\big(e^{-2\pi i\frac{x}{\go_3}}\big|e^{-2\pi i\frac{\go_1}{\go_3}},e^{-2\pi i\frac{\go_2}{\go_3}}\big)_{\infty}~.\nn
\eea
where $B_{3,3}$ is the Bernoulli polynomial defined as
\bea
B_{3,3}(z|\go_1,\go_2,\go_3)&=&\frac{z^3}{\go_1\go_2\go_3}-\frac32\frac{\go_1+\go_2+\go_3}{\go_1\go_2\go_3}z^2
+\frac{\go_1^2+\go_2^2+\go_3^2+3\go_1\go_2+3\go_2\go_3+3\go_3\go_1}{2\go_1\go_2\go_3}z\nn\\
&&-\frac{(\go_1+\go_2+\go_3)(\go_1\go_2+\go_2\go_3+\go_3\go_1)}{4\go_1\go_2\go_3}~.\label{bernoulli}\eea
These polynomials satisfy
\bea
B_{3,3}(z+\go_2|\go_1,\go_2,\go_3)=B_{3,3}(z|\go_1,-\go_2,\go_3)~.\label{bernoulli_prop}\eea

By comparing the two equivalent factorisations, one gets
\bea
e^{\frac{\pi i}{3}B_{3,3}(x|\go_1,\cdots,\go_3)}=\frac{\big(e^{2\pi i\frac{x}{\go_1}}\big|e^{2\pi i\frac{\go_2}{\go_1}},e^{2\pi i\frac{\go_3}{\go_1}}\big)_{\infty}}{\big(e^{-2\pi i\frac{x}{\go_1}}\big|e^{-2\pi i\frac{\go_2}{\go_1}},e^{-2\pi i\frac{\go_3}{\go_1}}\big)_{\infty}}\cdotp(\textrm{cyc perm in }\go_{1,2,3})~.\nn\eea
So one may also write the factorisation as
\bea
 &&S_3(x|\go_1,\cdots,\go_r) \label{factorisation_sine_sara} \\
 &&=\Big(\big(e^{2\pi i\frac{x}{\go_1}}\big|e^{2\pi i\frac{\go_2}{\go_1}},e^{2\pi i\frac{\go_3}{\go_1}}\big)_{\infty}\cdotp\big(e^{-2\pi i\frac{x}{\go_1}}\big|e^{-2\pi i\frac{\go_2}{\go_1}},e^{-2\pi i\frac{\go_3}{\go_1}}\big)_{\infty}\Big)^{1/2}(\textrm{cyc perm in }\go_{1,2,3})\nn\eea
without the Bernoulli polynomial but at the cost of having a square root. One can use (\ref{app-relations-1}) to rewrite the above as
\bea
S_3(x|\go_1,\cdots,\go_r)=\Big(\big(e^{2\pi i\frac{x}{\go_1}}\big|e^{2\pi i\frac{\go_2}{\go_1}},e^{2\pi i\frac{\go_3}{\go_1}}\big)_{\infty}\cdotp\big(x\to -x+\go_2+\go_3\big)\Big)^{1/2}\cdotp(\textrm{cyc perm in }\go_{1,2,3})~.\nn\eea

\subsection{A Lemma Concerning the Special Function}\label{sec_Lemma}

In this section we prove a useful identity, which may be of independent interest. To recapitulate the problem, one divides a 2-plane into a number of wedges with separating lines $\ell_i$ of rational slope. Assume that the normals (counter clockwise pointing) of every two neighbouring lines form an $SL(2,\BB{Z})$ basis, i.e. $\det[\vec u_i,\vec u_{i+1}]=\vec u_i\times \vec u_{i+1}=1$ for all $i$. Let $\vec r$ be a generic 2-vector in the sense that its imaginary part has irrational slope. We will prove
\bea \prod_k\big(x\big|\vec r\times \vec u_k,-\vec r\times \vec u_{k+1}\big)=1-e^{2\pi ix}.\label{abstruse}\eea
First, we remind the reader that we are using the short hand (\ref{short_hand}). Moreover, one has $\im(\vec r\times \vec u_i)\neq0$ for all $i$, so the special function above is well defined. The following is a direct proof, but it is also possible to prove this identity using (\ref{modular_new}) plus an induction similar to the one used when proving $c=-12$ at the end of section \ref{sec_act_cal}, which we leave to the reader.
\begin{proof}
Consider the line on $\BB{R}^2$ perpendicular to $\im\vec r$, since $\im\vec r$ is chosen generic, this line does not land on any integral points. We will only be interested in the four $\vec u$'s next to this line, see Figure \ref{fig_big_cancel}.
\begin{figure}[h]
\begin{center}
\begin{tikzpicture}[scale=1]
\draw [->] (0,-2) -- (0,2) node[left] {\small$\im\vec r$};
\draw [-,dashed] (-2,0) -- (2,0);

\draw [->,blue] (0,0) -- (1.2,1.2) -- (0.8,1.6) node [above] {\small$\vec u_{i+1}$};
\node at (.6,.6) [below] {\small$\ell_{i+1}$};

\draw [->,blue] (0,0) -- (1.2,-1.2) node [right] {\small$u_i$} -- (1.6,-0.8);
\node at (.6,-.6) [below] {\small$\ell_i$};


\draw [->,blue] (0,0) -- (-1.2,0.4) -- (-1.4,-0.2) node [left] {\small$\vec u_j$};
\draw [->,blue] (0,0) -- (-1.2,-0.4) -- (-1.0,-1.0) node [left] {\small$\vec u_{j+1}$};
\end{tikzpicture}\caption{The black line is $\im\vec r$, for the lines $\ell_j$ above the dotted line, its normal satisfies $\im\vec r\times \vec u_j>0$}\label{fig_big_cancel}
\end{center}
\end{figure}
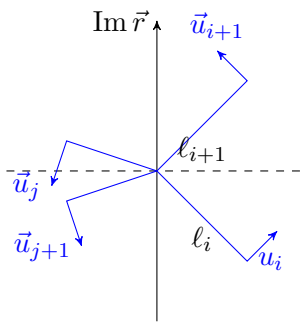

Let $\vec w,\vec v$ be the normals of two lines (ordered counterclockwise), and assume first $\im \vec r\cdotp\vec w>0,~\im \vec r\cdotp\vec v>0$, consider the infinite product
\bea P^{++}&=&\prod_{\vec n\cdotp\vec w>0;\vec n\cdotp \vec v\leq0}\big(1-e^{2\pi ix}\exp2\pi i(\vec n\cdotp\vec r)\big)\nn\\
&=&\prod_{\vec n\cdotp\vec w>0;\vec n\cdotp \vec v\leq0}\big(1-e^{2\pi ix}\exp2\pi i((\vec n\cdotp \vec w)(\vec r\times\vec v)-(\vec n\cdotp\vec v)(\vec r\times\vec w))\big)\nn\\
&=&\prod_{i,j\geq0}\big(1-e^{2\pi ix}\exp2\pi i((j+1)(\vec r\times\vec v)+i(\vec r\times \vec w)\big)=(x|\vec r\times\vec w,-\vec r\times\vec v)^{-1},\nn\eea
Similarly for $\im\vec r\times \vec w<0,~\im\vec r\times\vec v<0$,
\bea P^{--}&=&\prod_{\vec n\cdotp\vec w\geq0;\vec n\cdotp \vec v<0}\big(1-e^{2\pi ix}\exp2\pi i(-\vec n\cdotp\vec r)\big)=(x|\vec r\times u,-\vec r\times v)^{-1},\nn\eea
and for $\im\vec r\times\vec w>0,~\im\vec r\times \vec v<0$,
\bea P^{+-}&=&\prod_{n\cdotp \vec w\leq0;n\cdotp\vec v\leq0}\big(1-e^{2\pi ix}\exp2\pi i(\vec n\cdotp\vec r)\big)=(x|\vec r\times u,-\vec r\times v),\nn\eea
and finally if $\im\vec r\times \vec w<0,~\im\vec r\times \vec v>0$
\bea P^{-+}&=&\prod_{n\cdotp \vec w>0;n\cdotp \vec v>0}\big(1-e^{2\pi ix}\exp2\pi i(\vec n\cdotp\vec r)\big)=(x|\vec r\times \vec w,-\vec r\times\vec v).\nn\eea
With these preparations, we can finish the proof. The product from $(x|\vec r\times \vec u_{i+1},-\vec r\times \vec u_{i+2})$ to $(x|\vec r\times \vec u_{j-1},-\vec r\times\vec u_j)$ can be combined into a single product
\bea P_{(\vec u_{i+1},\vec u_j]}=\prod_{n\cdotp \vec u_{i+1}>0;n\cdotp \vec u_j\leq0}\big(1-e^{2\pi ix}\exp2\pi i(\vec n\cdotp\vec r)\big)^{-1}.\nn\eea
Similarly the factors from $(x|\vec r\times \vec u_{j+1},-\vec r\times \vec u_{j+2})$ to $(x|\vec r\times \vec u_{i-1},-\vec r\times \vec u_i)$
\bea P_{[\vec u_{j+1},\vec u_i)}=\prod_{n\cdotp \vec u_{j+1}\geq0;n\cdotp \vec u_i<0}\big(1-e^{2\pi ix}\exp2\pi i(-\vec n\cdotp\vec r)\big)^{-1}.\nn\eea
The factor $(x|e^{2\pi i(\vec r\times v_j)},e^{-2\pi i(\vec r\times v_{j+1})})$ can be written as
\bea P_{[-\vec u_{j+1},\vec u_j]}=\prod_{n\cdotp \vec u_j\leq0;n\cdotp \vec u_{j+1}\leq0}\big(1-e^{2\pi ix}\exp2\pi i(\vec n\cdotp\vec r)\big)=\prod_{n\cdotp \vec u_j\leq0;n\cdotp (-\vec u_{j+1})\geq0}\big(1-e^{2\pi ix}\exp2\pi i(\vec n\cdotp\vec r)\big),\nn\eea
the situation is depicted as in Figure \ref{fig_big_cancel_I}, that is, one flips $\vec u_{j+1}$ so that both $\vec u_j$ and $-\vec u_{j+1}$ stays above the dotted line.
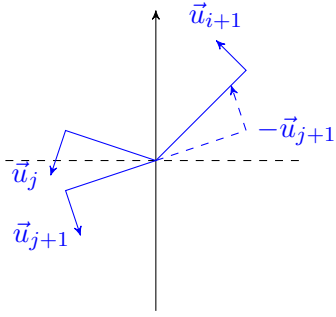
\begin{figure}[h]
\begin{center}
\begin{tikzpicture}[scale=1]
\draw [-,dashed] (-2,0) -- (2,0);
\draw [->] (0,-2) -- (0,2);

\draw [->,blue] (0,0) -- (1.2,1.2) -- (0.8,1.6) node [above] {\small$\vec u_{i+1}$};
\draw [->,blue,dashed] (0,0) -- (1.2,0.4) node [right] {\small$-\vec u_{j+1}$} -- (1.0,1.0);
\draw [->,blue] (0,0) -- (-1.2,0.4) -- (-1.4,-0.2) node [left] {\small$\vec u_j$};
\draw [->,blue] (0,0) -- (-1.2,-0.4) -- (-1.0,-1.0)  node [left] {\small$\vec u_{j+1}$};
\end{tikzpicture}\caption{One flips $\vec u_{j+1}$, and the product is now between $-\vec u_{j+1}$ and $\vec u_j$.}\label{fig_big_cancel_I}
\end{center}
\end{figure}
Then the combination
\bea P_{(\vec u_{i+1},\vec u_j]}P_{[-\vec u_{j+1},\vec u_j]}=\prod_{n\cdotp (-\vec u_{j+1})\geq0;n\cdotp \vec u_{i+1}\leq0}\big(1-e^{2\pi ix}\exp2\pi i(\vec n\cdotp\vec r)\big)=P_{[-\vec u_{j+1},\vec u_{i+1}]}.\label{detail_I}\eea
For the remaining factor $(x|e^{2\pi i(\vec r\times \vec u_i)},e^{-2\pi i(\vec r\times \vec u_{i+1})})$, consider the Figure \ref{fig_big_cancel_II}
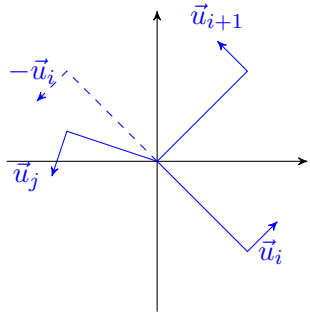
\begin{figure}[h]
\begin{center}
\begin{tikzpicture}[scale=1]
\draw [->] (-2,0) -- (2,0);
\draw [->] (0,-2) -- (0,2);

\draw [->,blue] (0,0) -- (1.2,1.2) -- (0.8,1.6) node [above] {\small$\vec u_{i+1}$};
\draw [->,blue] (0,0) -- (1.2,-1.2) node [right] {\small$\vec u_i$} -- (1.6,-0.8);
\draw [->,blue,dashed] (0,0) -- (-1.2,1.2) node [left] {\small$-\vec u_i$} -- (-1.6,0.8);


\draw [->,blue] (0,0) -- (-1.2,0.4) -- (-1.4,-0.2) node [left] {\small$\vec u_j$};
\end{tikzpicture}\caption{One flips $\vec u_i$, and the product is now between $\vec u_{i+1}$ and $-\vec u_i$.}\label{fig_big_cancel_II}
\end{center}
\end{figure}
and we get the contribution
\bea P_{(\vec u_{i+1},-\vec u_i)}=\prod_{\vec n\cdotp (-\vec u_i)<0;\vec n\cdotp \vec u_{i+1}>0}\big(1-x\exp2\pi i(\vec n\cdotp\vec r)\big)
=\prod_{\vec n\cdotp \vec u_i<0;\vec n\cdotp(-\vec u_{i+1})>0}\big(1-e^{2\pi ix}\exp2\pi i(-\vec n\cdotp\vec r)\big)~.\nn\eea
So the combination
\bea P_{[\vec u_{j+1},\vec u_i)}P_{(\vec u_{i+1},-\vec u_i)}&=&\prod_{\vec n\cdotp \vec u_{j+1}\geq0;\vec n\cdotp (-\vec u_{i+1})\leq0}\big(1-e^{2\pi ix}\exp2\pi i(-\vec n\cdotp\vec r)\big)^{-1}\cdotp(1-e^{2\pi ix})\nn\\
&=&\prod_{\vec n\cdotp(-\vec u_{j+1})\geq0;\vec n\cdotp \vec u_{i+1}\leq0}\big(1-e^{2\pi ix}\exp2\pi i(\vec n\cdotp\vec r)\big)^{-1}\cdotp(1-e^{2\pi ix})~.\nn\eea
Note that when one combines the two sums in two wedges, extra care is needed for the origin, this is the reason one has an extra $(1-e^{2\pi ix})$ factor above.
What we get here cancels the $P_{[-\vec u_{j+1},\vec u_{i+1}]}$ term from (\ref{detail_I}), leaving us with the factor $(1-e^{2\pi ix})$.
We have proved the cancellation assuming the particular arrangement of the four lines $\vec u_{i,i+1}$ $\vec u_{j,j+1}$ as in Figure
\ref{fig_big_cancel}, if they are arranged in a different relative position, the proof still goes through with only minor modifications \qed

\end{proof}

\section{A More Convenient formulation of the Good Cone Condition}\label{sec_childsplay}

The original goodness condition of a cone given by Lerman is the following, at every codimension-$k$ face, the $k$-normals $\vec v_{i_1},\cdots\vec v_{i_k}$ satisfies
\bea \textrm{span}_{\BB{R}}\bra\vec v_{i_1},\cdots\vec v_{i_k}\ket\cap\BB{Z}^{\tt m}=\textrm{span}_{\BB{Z}}\bra\vec v_{i_1},\cdots\vec v_{i_k}\ket~.\label{Lerman}\eea
This condition is equivalent to saying that $\{\vec v_{i_1},\cdots\vec v_{i_k}\}$ can be completed into an $SL({\tt m},\BB{Z})$-matrix. To see this, it is enough to consider ${\tt m}=3$.

At a codimension 1 face, we just have one normal, call it $\vec v$. For (\ref{Lerman}) to be true $\vec v$ must be primitive. This is also sufficient, indeed, suppose $\vec v=[p,q,r]$, $\gcd(p,q,r)=1$, there exist two integers $s,t$ such that $sq-tp=\gcd(p,q)$ ($s,\,t$ can be found using Euclid's algorithm). Consider the $SL(3,\BB{Z})$ matrix
\bea A=\left(
         \begin{array}{ccc}
           \bar q & -\bar p & 0 \\
           -t & s & 0 \\
           0 & 0 & 1 \\
         \end{array}\right),~~~\bar p=p/\gcd(p,q),~~\bar q=q/\gcd(p,q).\nn\eea
Clearly $A\vec v=[0,\gcd(p,q),r]$. Now since $\gcd(\gcd(p,q),r)=1$, one can find another $SL(3,\BB{Z})$ matrix $A'$ such that $A'A\vec v=[0,0,1]$.
Hence $A'A\vec v$ satisfies (\ref{Lerman}), and so $\vec v$ also does. From this argument, we also see that $\vec v$ can be completed into an $SL(3,\BB{Z})$ matrix. The above argument is quite a useful one, we restate it as, for any vector $\vec v$ of dimension $\tt m$, one can always find an $SL({\tt m},\BB{Z})$ matrix $A$ so that $A\vec v=[\gcd(\vec v),0,\cdots,0]$.

Now proceed to the codimension 2 face, which is the intersection of two codimension 1 faces with primitive normals $\vec u$, $\vec v$. One can find an $SL(3,\BB{Z})$ matrix to put $\vec v$ into $[0,0,1]$, denote by $w=A\vec u$. The span of $A\vec u,\,A\vec v$ is the same as the span of $[0,0,1]$ and $[w^1,w^2,0]$, showing that $\gcd(w^1,w^2)=1$ if (\ref{Lerman}) is to be satisfied. Then another $SL(3,\mathbb{Z})$ transformation can put $[\vec u,\vec v]$ into
\bea [\vec u,\vec v]\to\left(
       \begin{array}{cc}
         0 & 0 \\
         1 & 0 \\
         * & 1 \\
       \end{array}\right),\nn\eea
which can obviously be completed into an $SL(3,\BB{Z})$ matrix.
With minor modifications, the proof extends to higher dimensions as well.

\providecommand{\href}[2]{#2}\begingroup\raggedright\endgroup


\begin{thebibliography}{10}

\bibitem{Pestun:2007rz}
V.~Pestun, ``{Localization of gauge theory on a four-sphere and supersymmetric
  Wilson loops},'' \href{http://dx.doi.org/10.1007/s00220-012-1485-0}{{\em
  Commun.Math.Phys.} {\bfseries 313} (2012) 71--129},
\href{http://arxiv.org/abs/0712.2824}{{\ttfamily arXiv:0712.2824 [hep-th]}}.

\bibitem{Nekrasov:2003vi}
N.~Nekrasov, ``{Localizing gauge theories},''
{\em XIVth International Congress on Mathematical Physics} (2003) 645--654.

\bibitem{BoyerGalicki}
C.~P. Boyer and K.~Galicki, {\em Sasakian Geometry}.
\newblock Oxford University Press, USA, 2008.

\bibitem{2010arXiv1004.2461S}
J.~{Sparks}, ``{Sasaki-Einstein Manifolds},'' {\em ArXiv e-prints} (Apr., 2010)
  , \href{http://arxiv.org/abs/1004.2461}{{\ttfamily arXiv:1004.2461
  [math.DG]}}.

\bibitem{Qiu:2013pta}
J.~Qiu and M.~Zabzine, ``{5D Super Yang-Mills on $Y^{p,q}$ Sasaki-Einstein
  manifolds},''
\href{http://arxiv.org/abs/1307.3149}{{\ttfamily arXiv:1307.3149 [hep-th]}}.

\bibitem{Qiu:2013aga}
J.~Qiu and M.~Zabzine, ``{Factorization of 5D super Yang-Mills on $Y^{p,q}$
  spaces},'' \href{http://dx.doi.org/10.1103/PhysRevD.89.065040}{{\em
  Phys.Rev.} {\bfseries D89} (2014) 065040},
\href{http://arxiv.org/abs/1312.3475}{{\ttfamily arXiv:1312.3475 [hep-th]}}.

\bibitem{Nekrasov:2002qd}
N.~A. Nekrasov, ``{Seiberg-Witten prepotential from instanton counting},'' {\em
  Adv.Theor.Math.Phys.} {\bfseries 7} (2004) 831--864,
\href{http://arxiv.org/abs/hep-th/0206161}{{\ttfamily arXiv:hep-th/0206161
  [hep-th]}}.

\bibitem{Nekrasov:2003rj}
N.~Nekrasov and A.~Okounkov, ``{Seiberg-Witten theory and random partitions},''
\href{http://arxiv.org/abs/hep-th/0306238}{{\ttfamily arXiv:hep-th/0306238
  [hep-th]}}.

\bibitem{Nieri:2013vba}
F.~Nieri, S.~Pasquetti, F.~Passerini, and A.~Torrielli, ``{5D partition
  functions, q-Virasoro systems and integrable spin-chains},''
\href{http://arxiv.org/abs/1312.1294}{{\ttfamily arXiv:1312.1294 [hep-th]}}.

\bibitem{Nieri:2013yra}
F.~Nieri, S.~Pasquetti, and F.~Passerini, ``{3d 5d gauge theory partition
  functions as q-deformed CFT correlators},''
\href{http://arxiv.org/abs/1303.2626}{{\ttfamily arXiv:1303.2626 [hep-th]}}.

\bibitem{Gauntlett:2004yd}
J.~P. Gauntlett, D.~Martelli, J.~Sparks, and D.~Waldram, ``{Sasaki-Einstein
  metrics on $S^2 \times S^3$},''
  \href{http://dx.doi.org/10.4310/ATMP.2004.v8.n4.a3}{{\em
  Adv.Theor.Math.Phys.} {\bfseries 8} (2004) 711--734},
\href{http://arxiv.org/abs/hep-th/0403002}{{\ttfamily arXiv:hep-th/0403002
  [hep-th]}}.

\bibitem{Cvetic:2005ft}
M.~Cvetic, H.~Lu, D.~N. Page, and C.~Pope, ``{New Einstein-Sasaki spaces in
  five and higher dimensions},''
  \href{http://dx.doi.org/10.1103/PhysRevLett.95.071101}{{\em Phys.Rev.Lett.}
  {\bfseries 95} (2005) 071101},
\href{http://arxiv.org/abs/hep-th/0504225}{{\ttfamily arXiv:hep-th/0504225
  [hep-th]}}.

\bibitem{2001math......7201L}
E.~Lerman, ``Contact toric manifolds,'' {\em J. Symplectic Geom.} {\bfseries 1}
  no.~4, (2002) 659--828, \href{http://arxiv.org/abs/math/0107201}{{\ttfamily
  arXiv:math/0107201 [math]}}.
  \url{http://projecteuclid.org/getRecord?id=euclid.jsg/1092749569}.

\bibitem{Delzant1988}
T.~Delzant, ``Hamiltoniens périodiques et images convexes de l'application
  moment,'' {\em Bulletin de la Société Mathématique de France} {\bfseries
  116} no.~3, (1988) 315--339. \url{http://eudml.org/doc/87558}.

\bibitem{Kallen:2012cs}
J.~K\"all\'en and M.~Zabzine, ``{Twisted supersymmetric 5D Yang-Mills theory and
  contact geometry},'' \href{http://dx.doi.org/10.1007/JHEP05(2012)125}{{\em
  JHEP} {\bfseries 1205} (2012) 125},
\href{http://arxiv.org/abs/1202.1956}{{\ttfamily arXiv:1202.1956 [hep-th]}}.

\bibitem{Kallen:2012va}
J.~K\"all\'en, J.~Qiu, and M.~Zabzine, ``{The perturbative partition function of
  supersymmetric 5D Yang-Mills theory with matter on the five-sphere},''
  \href{http://dx.doi.org/10.1007/JHEP08(2012)157}{{\em JHEP} {\bfseries 1208}
  (2012) 157},
\href{http://arxiv.org/abs/1206.6008}{{\ttfamily arXiv:1206.6008 [hep-th]}}.

\bibitem{HosomichiSeongTerashima}
K.~Hosomichi, R.-K. Seong, and S.~Terashima, ``{Supersymmetric Gauge Theories
  on the Five-Sphere},''
\href{http://arxiv.org/abs/1203.0371}{{\ttfamily arXiv:1203.0371 [hep-th]}}.

\bibitem{Imamura:2012xg}
Y.~Imamura, ``{Supersymmetric theories on squashed five-sphere},''
\href{http://arxiv.org/abs/1209.0561}{{\ttfamily arXiv:1209.0561 [hep-th]}}.

\bibitem{Imamura:2012bm}
Y.~Imamura, ``{Perturbative partition function for squashed $S^5$},''
\href{http://arxiv.org/abs/1210.6308}{{\ttfamily arXiv:1210.6308 [hep-th]}}.

\bibitem{Kim:2012ava}
H.-C. Kim and S.~Kim, ``{M5-branes from gauge theories on the 5-sphere},''
  \href{http://dx.doi.org/10.1007/JHEP05(2013)144}{{\em JHEP} {\bfseries 1305}
  (2013) 144},
\href{http://arxiv.org/abs/1206.6339}{{\ttfamily arXiv:1206.6339 [hep-th]}}.

\bibitem{Lockhart:2012vp}
G.~Lockhart and C.~Vafa, ``{Superconformal Partition Functions and
  Non-perturbative Topological Strings},''
\href{http://arxiv.org/abs/1210.5909}{{\ttfamily arXiv:1210.5909 [hep-th]}}.

\bibitem{Kim:2012qf}
H.-C. Kim, J.~Kim, and S.~Kim, ``{Instantons on the 5-sphere and M5-branes},''
\href{http://arxiv.org/abs/1211.0144}{{\ttfamily arXiv:1211.0144 [hep-th]}}.

\bibitem{Ellip_Ope_Cpct_Grp}
M.~F. Atiyah, {\em {Elliptic operators and compact groups}}, vol.~401.
\newblock Springer-Verlag, Berlin, 1974.

\bibitem{Schmude:2014lfa}
J.~Schmude, ``{Localisation on Sasaki-Einstein manifolds from holomophic
  functions on the cone},''
\href{http://arxiv.org/abs/1401.3266}{{\ttfamily arXiv:1401.3266 [hep-th]}}.

\bibitem{Schmude:2013dua}
J.~Schmude, ``{Laplace operators on Sasaki-Einstein manifolds},''
\href{http://arxiv.org/abs/1308.1027}{{\ttfamily arXiv:1308.1027 [hep-th]}}.

\bibitem{Eager:2012hx}
R.~Eager, J.~Schmude, and Y.~Tachikawa, ``{Superconformal Indices,
  Sasaki-Einstein Manifolds, and Cyclic Homologies},''
\href{http://arxiv.org/abs/1207.0573}{{\ttfamily arXiv:1207.0573 [hep-th]}}.

\bibitem{Benvenuti:2006qr}
S.~Benvenuti, B.~Feng, A.~Hanany, and Y.-H. He, ``{Counting BPS Operators in
  Gauge Theories: Quivers, Syzygies and Plethystics},''
  \href{http://dx.doi.org/10.1088/1126-6708/2007/11/050}{{\em JHEP} {\bfseries
  0711} (2007) 050},
\href{http://arxiv.org/abs/hep-th/0608050}{{\ttfamily arXiv:hep-th/0608050
  [hep-th]}}.

\bibitem{Fulton}
W.~Fulton, {\em {Introduction to toric varieties}}.
\newblock No.~131 in Annals of mathematics studies. Princeton University Press,
  1993.

\bibitem{MR2101221}
A.~Narukawa, ``The modular properties and the integral representations of the
  multiple elliptic gamma functions,''
  \href{http://dx.doi.org/10.1016/j.aim.2003.11.009}{{\em Adv. Math.}
  {\bfseries 189} no.~2, (2004) 247--267}.
  \url{http://dx.doi.org/10.1016/j.aim.2003.11.009}.

\bibitem{Martelli:2005tp}
D.~Martelli, J.~Sparks, and S.-T. Yau, ``{The Geometric dual of a-maximisation
  for Toric Sasaki-Einstein manifolds},''
  \href{http://dx.doi.org/10.1007/s00220-006-0087-0}{{\em Commun.Math.Phys.}
  {\bfseries 268} (2006) 39--65},
\href{http://arxiv.org/abs/hep-th/0503183}{{\ttfamily arXiv:hep-th/0503183
  [hep-th]}}.

\bibitem{Kallen:2012zn}
J.~K\"all\'en, J.~Minahan, A.~Nedelin, and M.~Zabzine, ``{$N^3$-behavior from 5D
  Yang-Mills theory},'' \href{http://dx.doi.org/10.1007/JHEP10(2012)184}{{\em
  JHEP} {\bfseries 1210} (2012) 184},
\href{http://arxiv.org/abs/1207.3763}{{\ttfamily arXiv:1207.3763 [hep-th]}}.

\end{thebibliography}
\end{document}